\documentclass{article}
\usepackage{graphicx,amsfonts,amsmath,mathrsfs,amssymb,amsthm,url,color}
\usepackage{fancyhdr,indentfirst,bm,enumerate,natbib}
\usepackage{tikz}
\usepackage[colorlinks=true,citecolor=blue]{hyperref}
\usepackage{dsfont}
\def\Fbar{ {\overline F}}         
          \def\Gbar{ {\overline G}}
\def\oo{\infty}                   \def\d{\,\mathrm{d}}
\def\lm{\lambda}                  
                 \def\ep{{\epsilon}}
\def\dfrac{\displaystyle \frac}
\newcommand{\id}{\mathds{1}}

\newcommand{\E}{\mathbb{E}}
\newcommand{\R}{\mathbb{R}}

\newcommand{\VaR}{\mathrm{VaR}}
\newcommand{\TVaR}{\mathrm{TVaR}}
\newcommand{\RVaR}{\mathrm{RVaR}}
\newcommand{\essinf}{\mathrm{ess\mbox{-}inf}}
\newcommand{\esssup}{\mathrm{ess\mbox{-}sup}}
\renewcommand{\P}{\mathbb{P}}
\renewcommand{\(}{\left (}
\renewcommand{\)}{\right )}
\renewcommand{\[}{\left [}
\renewcommand{\]}{\right ]}
\newtheorem{theorem}{Theorem}[section]

\newtheorem{lemma}[theorem]{Lemma}
\newtheorem{proposition}[theorem]{Proposition}
\newtheorem{definition}{Definition}[section]

\newtheorem{remark}[theorem]{Remark}

\numberwithin{equation}{section}
\numberwithin{theorem}{section}

\renewcommand{\cite}{\citet}
\allowdisplaybreaks

\def\blue{\color{blue}}

\begin{document}
	
\baselineskip=18pt
	
\title{Extreme-case Range Value-at-Risk under Increasing Failure Rate}
	
\author{ Yuting Su\qquad Taizhong Hu\qquad Zhenfeng Zou\footnote{Corresponding author. \newline
 \hbox{\quad\ }  E-mail addresses: {\blue syt20020224@mail.ustc.edu.cn} (Y. Su),  {\blue thu@ustc.edu.cn} (T. Hu), {\blue zfzou@ustc.edu.cn} (Z. Zou).  }   \\[10pt]
	School of Management, University of Science and Technology of China,\\
	Hefei, Anhui 230026, China   	}
	
\date{June, 2025}
	
\maketitle
	
\begin{abstract}
		
The extreme cases of risk measures, when considered within the context of distributional ambiguity, provide significant guidance for practitioners specializing in risk management of quantitative finance and insurance. In contrast to the findings of preceding studies, we focus on the study of extreme-case risk measure under distributional ambiguity with the property of increasing failure rate (IFR). The extreme-case range Value-at-Risk under distributional uncertainty, consisting of given mean and/or variance of distributions with IFR, is provided. The specific characteristics of extreme-case distributions under these constraints have been characterized, a crucial step for numerical simulations. We then apply our main results to stop-loss and limited loss random variables under distributional uncertainty with IFR.
		
\medskip
		
\noindent \textbf{Mathematics Subject Classifications (2000)}: 90C15; 90C26		\medskip
		
\noindent \textbf{Keywords}: Distributional ambiguity; Moment uncertainity; Increasing failure rate; Risk measure; Range Value-at-Risk
		
\end{abstract}

\section{Introduction}
	
The accurate measurement of the risk of a loss variable is a significant challenge for those engaged in risk management. Over the past few decades, researchers have typically assumed that the distribution of the loss variable is known and then calculated the capital requirement to withstand that risk based on that distribution. For example, the well-known \emph{Value at Risk} (VaR), a standard risk measure in the insurance regulatory framework of Solvency II, determines the risk level of a loss variable by measuring the quantile of its distribution. Another standard risk measure in the Basel Accords as well as the Swiss Solvency Test, \emph{Tail Value at Risk} (TVaR), is gradually becoming the popular calculation method to replace VaR in measuring risk. It refers to the average loss value above a certain quantile of a loss variable and is also termed as \emph{Average Value-at-Risk} (AVaR), \emph{Expected Shortfall} (ES) and \emph{Conditional Value-at-Risk} (CVaR) in different contexts. In fact, VaR and TVaR belong to a more general class of risk measure, \emph{Range Value at Risk} (RVaR), which was proposed by \cite{CDS10} for robustness considerations. However, in most practical situations, it is difficult to obtain the exact distribution of a loss variable. Therefore, risk analysis under distributional ambiguity naturally appears in the literature.
	
The term ``distributional ambiguity'' is used to describe the situation in which the probability distribution of the loss variable under consideration by risk practitioners is unknown or partially known. Practitioners usually infer the characteristics of the distribution (including moments, unimodality or symmetry) based on limited historical data, and then construct an uncertainty set to capture distributional uncertainty to determine the risk level of possible extreme-case, including worst-case and best-case. A number of recent studies have demonstrated the importance of distributional ambiguity in the area of risk analysis. It would appear that the worst-case VaR, as determined by the information provided by the first two moments, was initially considered by \cite{EOO03}. They provided a closed-form expression and characterized the worst-case distribution that reaches the worst-case scenario. Subsequently, under the uncertainty set given the information of the first two moments, the closed-form expressions for the worst-case TVaR and RVaR were provided by \cite{CHZ11} and \cite{LSWY18}, respectively. Inspired by the aforementioned studies,  \cite{Li18} studied the case of worst-case spectral risk measures and worst-case law-invariant coherent risk measures and also obtained closed-form expressions when only the first two moments are known for the underlying distribution. Recently, based on the aforementioned uncertainty set for the underlying distribution, \cite{SZ23} and \cite{ZY25} provided the closed-form expressions for the worst-case distortion risk measures and distortion riskmetrics, respectively. Moreover, the extreme-case distortion risk measures were also considered by \cite{SZ24} under the assumption of knowledge of first-order and higher-order moments. Notably, \cite{PWW24} established a novel sufficient condition for closedness under concentration operators within  distributional ambiguity sets. This theoretical advancement enables the transformation of general non-convex distortion riskmetric into convex counterparts through the concave envelope of the distortion functions. Furthermore, they obtained the closed-form expressions for extreme-case distortion riskmetrics within a distributional ambiguity set consisting of first-order and higher-order moments.
	
A review of the existing literature reveals a number of critiques concerning the distributional ambiguity information with only moments. A major criticism is that, regardless of the form of the underlying probability distribution, the extreme-case distribution that can be reached in extreme-case is a discrete distribution with a finite number of points. To overcome this drawback, the existing literature also considers distributional ambiguity composed of moments, in conjunction with shape constraints. The introduction of shape constraints, however, gives rise to new theoretical challenges, which must be addressed in order to facilitate further progress in the field of risk analysis under distributional ambiguity. \cite{LSWY18} derived the partially closed-form expressions for the worst-case scenarios of RVaR by using construction method, in single and aggregate risk models with given mean and variance, as well as symmetry and/or unimodality of each risk. When the mean and variance of a given loss variable and unimodal property are considered, it can be seen that \cite{BKV20,BKV23} fully encapsulated the extreme-case of RVaR via stochastic order. Furthermore, these authors derived upper and lower bounds for RVaR of a unimodal random variable under the assumption of the mean, variance, symmetry, and a possibly bounded support in \cite{BKV22}. \cite{SZ23, SZ24} provided extreme-case distortion risk measures in the context of distributional ambiguity, constructed under first-order and second-order moments or higher-order moment, as well as symmetry, and characterized the extreme-case distributions when reaching extreme cases. In their study, \cite{ZBY24} presented a unified framework for extremal problems of distortion risk measures based on the first two moments along with additional shape information, such as the symmetry and/or unimodality properties of the underlying distributions.
	
Apart from the shape constraints previously discussed, namely unimodal and symmetry, the distributions of loss variables also exhibit other shape constraints, including \emph{increasing failure rate} (IFR) and so forth. The concept of a monotone failure rate, in particular the IFR, has constituted a pivotal element in the field of reliability theory since the early 1960s, see \cite{BMP63, BM64, BM65}. Moreover, IFR is of considerable importance in a number of other fields, including those of inventory management \citep{LP01} and contract theory \citep{DJ16}. It is well known that the IFR property is closed under sums of independent random variables, which may be one of the reasons why IFR is widely used in applications. It is evident that the literature on extreme-case for risk measures under IFR constraint is scarce, with the exception of \cite{CHJRZ21}, due to the non-convex problem of distributional ambiguity with IFR constraints. In that research, \cite{CHJRZ21} investigated discrete moment problems with IFR shape constraint and characterized optimal extreme point distributions through a reverse convex optimization approach.
	
In order to address this research gap in the use of IFR in risk analysis under distributional ambiguity, the present study will primarily consider extreme-case RVaR with the first two moments and IFR information of a given loss variable. The main contributions of this paper are as follows.
\begin{itemize}
   \item[(1)] In the context of first-order moment and IFR, the upper and lower bounds of VaR are determined by the survival functions' bounds in the existing literature (Theorem \ref{th-VaR}). The upper and lower bounds of the coherent distortion risk measures are established by convex order (Theorem \ref{th-coherent}).
		
  \item[(2)] The extreme-case RVaRs under the first-order moment and IFR are established by construction method (Theorems \ref{thm3-4} and \ref{thm3-7}).
		
  \item[(3)] The worst-case TVaR under the first two moments and IFR is given (Theorem \ref{thm-4-4}) as well as the extreme-case RVaRs are also presented (Theorems \ref{thm-4-5} and \ref{thm-4-6}).
		
  \item[(4)] The extreme-case stop-loss transformation and the extreme-case limited loss random variables, which are prevalent in (re)insurance applications, are also considered (Propositions \ref{pro5-1} and \ref{pro5-2}).
\end{itemize}
	
The presentation of this paper is structured as follows. In Section \ref{sec2}, we give the formal definition of the distortion risk measure, incorporating VaR, TVaR and RVaR, and then introduce the concept of IFR as well as several properties that play a pivotal role in the ensuing analysis. Under the first-order moment and IFR constraint, we provide the extreme-case of RVaR, while also characterizing extreme-case VaR and coherent distortion risk measures in Section \ref{sec3}. Section \ref{sec4} generalizes our upper and lower bounds for RVaR to the first two moments and IFR constraint ambiguity set. We then apply our main results to stop-loss and limited loss random variables under distributional uncertainty with IFR in Section \ref{sec5}. Section \ref{sec6} concludes this paper.

\section{Preliminaries}\label{sec2}

Consider a probability space $(\Omega, \mathscr{A}, \P)$. Denote by $\cal M$ the set of all distributions on $\R_+$, i.e., we assume that $F(0-) = 0$ holds throughout the article for $F \in {\cal M}$. Moreover, let $\Fbar(x) := 1-F(x)$ be the survival function of $F$. Then we can define the left-continuous version $\Fbar(x-) := \lim_{\ep \to 0} \Fbar(x-\ep)$  for $\Fbar$. As we known, VaR at the level $\alpha \in (0, 1)$ is defined as the quantile of $F \in {\cal M}$, which has left- and right-continuous versions, i.e.,
\begin{align*}
	\VaR_\alpha(F) & := F^{-1}(\alpha) = \inf\left\{x \in [0, \oo): F(x) \geq \alpha \right\},\\
	\VaR^+_\alpha(F) & := F^{-1+}(\alpha)= \sup\left\{x \in [0, \oo): F(x-) \leq \alpha \right\}.
\end{align*}
In addition, $\VaR_0(F) = \VaR_0^+(F) := \essinf (F)$ and $\VaR_1(F) = \VaR_1^+(F) := \esssup (F)$, where $\essinf (F)$ and $\esssup (F)$ denote the essential infimum and essential supremum of $F$, respectively. The distinction between the left- and right-continuous versions of VaR is of negligible significance in practice, and they are frequently indistinguishable. Nevertheless, in the event of extreme-case scenarios under distributional ambiguity, it is imperative to discern the subtle differences, as evidenced by \cite{Rusc82}. To overcome the major drawback of VaR, which is that it cannot provide enough information on the magnitude of losses in case of default, \cite{AT02} proposed TVaR, which is defined as
\begin{equation*}
		\TVaR_\alpha(F) := \frac{1}{1-\alpha} \int_\alpha^1 \VaR_u(F) \d u,\quad \alpha \in [0, 1),
\end{equation*}
and $\TVaR_1(F) := \esssup (F)$ for $\alpha = 1$. However, TVaR may result in a less robust risk measurement procedure compared to VaR. Therefore, RVaR was introduced by \cite{CDS10} as a robust risk measure, defined by
\begin{equation}\label{RVaR}
    \RVaR_{\alpha,\beta}(F) := \frac{1}{\beta-\alpha} \int_\alpha^\beta \VaR_u(F) \d u,\quad 0 \leq \alpha < \beta \leq 1.
\end{equation}
RVaR incorporates as special cases both the VaR and TVaR. To be concrete, RVaR reduces to TVaR when $\beta = 1$ and results in VaR when $\beta \downarrow \alpha$.
	
The above three risk measures VaR, TVaR and RVaR belong to distortion risk measures. The class risk measures were introduced by \cite{Wang96} in the actuarial literature. Denote by
$$
    {\cal H} := \left\{h: h\ \hbox{maps}\ [0, 1]\ \hbox{to}\ [0, 1], h\ \hbox{is\ increasing\ with}\ h(0) = 0\ \hbox{and}\ h(1) = 1\right\},
$$
where the term ``increasing" or ``decreasing" is in the weak sense throughout the paper. Then, distortion risk measure of the general distribution function $F$ is defined by $\rho_h$ via Choquet integrals \citep{Cho54}, i.e.,
\begin{equation*}
		\rho_h(F) = \int_0^\oo h\(\Fbar(x)\) \d x - \int_{-\oo}^0 \[1 - h\(\Fbar(x)\)\] \d x.
\end{equation*}
When $F(0-)=0$, $\rho_h(F)$ reduces to
\begin{equation*}
		\rho_h(F) = \int_0^\oo h\(\Fbar(x)\) \d x.
\end{equation*}
For more properties on distortion risk measures and distortion riskmetrics, see \cite{WWW20}.
	
In this paper, we attempt to derive the analytical solutions for the worst-case or best-case risk measures based on the partial information of the underlying distribution, i.e.,
\begin{equation*}
	\sup_{F \in {\cal P}} \rho(F)\ {\rm and}\ \inf_{F \in {\cal P}} \rho(F)
\end{equation*}
where $\rho$ is some risk measure, and $\cal P$ is the uncertainty set of distributions in $\cal M$ with IFR, given the mean or variance. Let's recall the definitions of IFR and DFR (decreasing failure rate).
	
\begin{definition}\label{IFR}
Given $F \in {\cal M}$ with density function $f$, denote by $\lm_F(x)$ the associated failure rate function of $F$, defined by
\begin{equation*}
	\lm_F(x) := \frac{f(x)}{\Fbar(x)},\quad x\in \left\{y\in\R_+: \Fbar(y)>0\right \},
\end{equation*}
Here, we use the convention that $\lm_F(x)=+\oo$ whenever $\Fbar(x)=0$.
\begin{itemize}
 \item[{\rm (1)}] $F$ is said to have an increasing failure rate {\rm (IFR)} if $\lm_F(x)$ is increasing in $x\in\R_+$;

 \item[{\rm (2)}] $F$ is said to have a decreasing failure rate {\rm (DFR)} if $\lm_F(x)$ is decreasing in $x\in\R_+$.
\end{itemize}
\end{definition}
	
As stated in the introduction, the concept of monotonic hazard rate has played an important role in reliability theory since the early 1960s, as evidenced by the works of \cite{BMP63, BM64, BM65}. It is a conventional understanding in reliability engineering that the lifetime of a component exhibits a strong inverse correlation with its failure rate. This relationship stems primarily from the effects of wear and tear. Distribution function $F \in {\cal M}$ with IFR has many desirable properties and characterizations. Below, we summarize the main properties in the following Proposition \ref{pro-2-1} that facilitate our subsequent analysis. The abbreviation ``$F$ is IFR" is used in place of the more complete ``$F$ has an increasing failure rate."
	
\begin{proposition}\label{pro-2-1}
If $F$ is {\rm IFR} with finite mean $\mu(F) := \int_0^\oo x \d F(x)$, then
\begin{enumerate}[{\rm (i)}]
  \item $F(x)$ is absolute continuous on $D_F$, where $D_F := [0, d)$ with $d := \sup\left\{t \geq 0: \Fbar(t) > 0\right\}$. Moreover, $F(0+) = F(0) = 0$;

  \item $\Lambda_F(x) := -\ln\Fbar(x)$ is a convex function in $x \in [0, \oo)$;

  \item for  any $c>0$ and $a>0$, the function $\Fbar(x) - c\exp\{-ax\}$ changes its sign at most twice with $-,+,-$ for $x \in [0, \oo)$;

  \item the total time on test transform function {\rm (TTT)} $\phi_F(u)$ is concave in $u\in [0,1]$, where
	\begin{equation} \label{eq-ttt}
		\phi_F(u) := \frac {1}{\mu(F)} \int_0^{F^{-1}(u)} \Fbar(x) \d x.
	\end{equation}
\end{enumerate}
\end{proposition}
	
In fact, statements (ii)-(iv) in Proposition 2.1 are equivalent. They can all be used to characterize the IFR property of a distribution function. For the equivalent characterization (iv), we refer to \cite{BC75}.

\section{Extreme-cases via mean}
\label{sec3}

In this section, when the mean of $F$ and $F$ being IFR are known, we consider an ambiguity set specified by
\begin{equation*}
	{\cal F}_\mu := \left\{F \in {\cal M}: F {\rm\ is\ IFR},\ \mu(F) = \mu \right\}.
\end{equation*}
The following questions will be addressed:
\begin{equation}
	\sup_{F \in {\cal F}_\mu} \rho(F),\ {\rm and}\ \inf_{F \in {\cal F}_\mu} \rho(F),
\end{equation}
where $\rho$ is a risk measure. The subsequent subsections will provide a concise exposition of the methodology for extreme-case risk measures in the context of distributional uncertainty ${\cal F}_\mu$. In Subsection \ref{subsec31}, the known bounds of the survival functions in the literature are utilized to provide representations of the upper and lower bounds of VaR. Convex order is then employed to delineate the bounds of the coherent distorted risk measure in Subsection \ref{subsec32}. The most salient conclusions about extreme-case RVaR are presented in Subsection \ref{subsec33}. Without loss of generality, we assume $\mu=1$ and denote ${\cal F} := {\cal F}_1$

\subsection{VaRs}\label{subsec31}

In this subsection, we consider $\rho = \VaR^+_\alpha$ or $\rho = \VaR_\alpha$ for $\alpha \in (0, 1)$. First, we recall some bounds for the survive function $\Fbar$ when $F \in {\cal F}$. These results are very useful for solving the case of VaRs. The upper bound is given blow:
\begin{equation}\label{sup-IFR}
	M(t) := \sup_{F \in {\cal F}} \Fbar(t-) = \left\{\begin{array}{ll} 1, & \hbox{for}\ t \in (0,1], \\
			\exp\{-wt\}, & \hbox{for}\ t > 1, \end{array} \right.
\end{equation}
where $w\in\R_+$ is determined by the equation $\exp\{-wt\} = 1 - w$. Moreover, the supremum in \eqref{sup-IFR} can be attainable. Specifically, the upper bound \eqref{sup-IFR} is achieved respectively by the \emph{Dirac delta distribution} at $1$, $\delta_1$, for $t\in (0,1]$, and $\Fbar_0(x) = \exp\{-wx\} \id_{\{x < t\}}$ with $w$ determined by $\mu(F_0) = 1$ for $t>1$. Here, $\delta_1$ is the distribution of a degenerate random variable $Z=1$. Besides, the lower bound has the form:
\begin{equation}\label{inf-IFR}
	m(t) := \inf_{F \in {\cal F}} \Fbar(t) = \left\{\begin{array}{ll} \exp\{-t\}, & \hbox{for}\ t \in (0,1), \\
		0, & \hbox{for}\ t \geq 1. \end{array} \right.
\end{equation}
Moreover, the infimum in \eqref{inf-IFR} can be attainable. In particular, this lower bound \eqref{inf-IFR} is achieved by the ${\rm Exp}(1)$ distribution for $t\in (0,1)$ and $\delta_1$ for $t\ge 1$, respectively, where ${\rm Exp}(1)$ is the exponential distribution with parameter $1$. A minor modification of \eqref{inf-IFR} is as follows:
$$
	m^\ast(t) :=\inf_{F\in {\cal F}} \Fbar(t-) =\left\{\begin{array}{ll}\exp\{-t\}, & \hbox{for}\ t\in (0,1],\\
		0, & \hbox{for}\ t > 1. \end{array} \right.
$$
For more details, we refer to \cite{BM64}. Now, we have the following proposition.
	
\begin{theorem} 		\label{th-VaR}
For $\alpha \in (0, 1)$, we have
\begin{align}
   \sup_{F \in {\cal F}} \VaR^+_\alpha (F) & = M^{-1}(1-\alpha):= \sup\left\{t\in\R_+: M(t)\geq 1-\alpha \right\},  \label{eq-2024-2}
\end{align}
and
\begin{align}
   \inf_{F\in {\cal F}} \VaR_\alpha (F) & =m^{-1}(1-\alpha):= \sup\left\{t\in\R_+: m(t)> 1-\alpha \right\}  \label{eq-2024-3} \\
   & = \left\{\begin{array}{ll} -\ln (1-\alpha), & {\rm for}\ \alpha<1-e^{-1},\\
       1, & {\rm for}\ \alpha\ge 1-e^{-1}. \end{array} \right.  \label{eq-2024-5}
\end{align}
\end{theorem}
	
\begin{proof}
By the definition of $M^{-1}(1-\alpha)$, we have
\begin{align*}
  M^{-1}(1-\alpha) &= \sup\left\{t \in\R_+: M(t) \geq 1 - \alpha \right\} \\
	 &= \sup\left\{t \in\R_+: \sup_{F \in {\cal F}} \Fbar(t-) \geq 1 - \alpha \right\} \\
	 &= \sup\left\{t \in\R_+: \inf_{F \in {\cal F}} F(t-) \leq \alpha \right\} \\
        &\ge  \sup_{F \in {\cal F}} \sup\left\{t\in\R_+:  F(t-) \leq \alpha \right\} = \sup_{F \in {\cal F}} \VaR^+_\alpha (F).
\end{align*}
For the converse statement, take any $t_1$ such that $t_1 > \sup_{F \in {\cal F}} \VaR^+_\alpha (F)$. There exists $t_2$ such that $t_1>t_2 > \VaR^+_\alpha (F)$ for any $F \in {\cal F}$, which implies that $\inf_{F \in {\cal F}} F(t_2-) \ge \alpha$, i.e., $M(t_2) \le 1-\alpha$. From \eqref{sup-IFR}, it is known that $M(t)$ is strictly decreasing on $[1, \oo)$. Thus, $M(t_1)<M(t_2)\le 1-\alpha$, implying $t_1 \ge M^{-1}(1-\alpha)$. It is easy to conclude that $M^{-1}(1-\alpha) \leq \sup_{F \in {\cal F}} \VaR^+_\alpha (F)$. This proves \eqref{eq-2024-2}.
		
To prove \eqref{eq-2024-3}, note that
\begin{align*}
	m^{-1}(1-\alpha) &= \sup\left\{t \in\R_+: m(t) > 1 - \alpha \right\} \\
		&= \sup\left\{t \in\R_+: \inf_{F \in {\cal F}} \Fbar(t) > 1 - \alpha \right\} \\
		&= \sup\left\{t \in\R_+: \sup_{F \in {\cal F}} F(t) < \alpha \right\} \\
		&\le \inf_{F \in {\cal F}} \sup\left\{t\in\R_+:  F(t) < \alpha \right\}
             = \inf_{F \in {\cal F}} \VaR_\alpha (F).
\end{align*}
For the converse statement, taking any $t_3, t_4$ such that $t_3 <t_4< \inf_{F \in {\cal F}} \VaR_\alpha (F)$, we have $t_4 < \VaR_\alpha (F)$ for any $F \in {\cal F}$. This implies $\inf_{F \in {\cal F}} \Fbar(t_4) \ge 1-\alpha$, i.e., $m(t_4) \ge 1-\alpha$. From \eqref{inf-IFR}, it is known that $m(t)$ is strictly decreasing on $(0, 1)$. Thus, $m(t_3)>1-\alpha$, implying $m^{-1}(1-\alpha)\ge t_3$. It is easy to conclude that $m^{-1}(1-\alpha) \geq \sup_{F \in {\cal F}} \VaR_\alpha (F)$. This proves \eqref{eq-2024-3}.

In view of \eqref{inf-IFR}, \eqref{eq-2024-5} follows by observing that $m^{-1}(1-\alpha)= \sup\{t\in (0,1): e^{-t}>1-\alpha\} = \sup\{t\in (0,1): t < -\ln (1-\alpha) \}$. This completes the proof of the theorem.
\end{proof}
	
\begin{remark}{\rm
For $r>0$, denote the $r$-th moment of a distribution $F\in {\cal M}$ by
$$
	\mu_r (F) := \int_0^\oo x^r \d F(x).
$$
For given $\mu_r\in \R_+$, define
\begin{equation*}
	{\cal K}_{\mu_r} = \left\{F \in {\cal M}: F {\rm\ is\ IFR},\ \mu_r(F)=\mu_r \right\}.
\end{equation*}
It is known from \cite{BM64} that \vspace*{-4pt}
\begin{itemize}
  \item for $r>0$,
	 \begin{equation}\label{sup-IFRr}
	      M_r(t) := \sup_{F \in {\cal K}_{\mu_r}} \Fbar(t-) = \left\{\begin{array}{ll} 1,
                & \hbox{for}\ t \leq \mu_r^{1/r}, \\[3pt]
				\exp\{-wt\}, & \hbox{for}\ t > \mu_r^{1/r}, \end{array} \right.
	\end{equation}
	where $w$ is determined by the equation $\int_0^t r z^{r-1} \exp\{-wz\} \d z = \mu_r$;
 \item for $r \geq 1$,
	\begin{equation}\label{inf-IFRr}
		m_r(t) := \inf_{F \in {\cal K}_{\mu_r}} \Fbar(t) = \left\{\begin{array}{ll}
             \exp\left\{-t {\big /} \lm_r^{1/r}\right \}, & \hbox{for}\ t < \mu_r^{1/r}, \\[3pt]
						0, & \hbox{for}\ t \geq \mu_r^{1/r}, \end{array} \right.
	\end{equation}
	where $\lm_r := \mu_r(F) / \Gamma(r+1)$.
\end{itemize}
Moreover, the supremum in \eqref{sup-IFRr} and the infimum in \eqref{inf-IFRr} can be attainable.
			
From the proof of Theorem \ref{th-VaR}, we know that these results also hold on ${\cal K}_{\mu_r}$, i.e.,
$$
     \sup_{F\in {\cal K}_{\mu_r}} \VaR^+_\alpha (F) =M_r^{-1}(1-\alpha):= \sup\left\{t \in [0, \oo): M_r(t) \geq 1- \alpha \right\}
$$
and
\begin{align*}
     \inf_{F \in {\cal K}_{\mu_r}} \VaR_\alpha (F) & = m_r^{-1}(1-\alpha) := \sup\left\{t \in [0, \oo): m_r(t) > 1- \alpha \right\} \\
     & = \left\{\begin{array}{ll} - \lm_r^{-1/r} \ln (1-\alpha), & {\rm for}\ \alpha<1-\exp\left\{-(\lm_r\mu_r)^{1/r}\right \},\\
       \mu_r^{1/r}, & {\rm for}\  \alpha \geq 1-\exp\left\{-(\lm_r\mu_r)^{1/r}\right \}. \end{array} \right.
\end{align*}
}
\end{remark}

\subsection{Coherent distortion risk measure}\label{subsec32}

In this subsection, we consider the case of coherent distortion risk measure $\rho_h$, i.e., the distortion function $h \in {\cal H}$ is concave. First, we introduce some notations. For any $F, G \in {\cal M}$, if there exist $t\in\R$ and $\delta > 0$ such that
\begin{equation*}
		F(x) < G(x),\ \forall x \in (t-\delta, t),\ {\rm and}\ F(x) > G(x),\ \forall x \in (t, t+\delta),
\end{equation*}
then we say $F$ up-crosses $G$ at $t$. If
\begin{equation*}
		F(x) > G(x),\ \forall x \in (t-\delta, t),\ {\rm and}\ F(x) < G(x),\ \forall x \in (t, t+\delta),
\end{equation*}
then we say $F$ down-crosses $G$ at $t$. Intuitively, when we say $F$ up-crosses (down-crosses) $G$, it means that $F$ crosses $G$ from below (above). Denote by $\E[\phi(F)] := \int_0^\oo \phi(x) \d F(x)$ for any measurable function $\phi: \R \to \R$ and $F \in {\cal M}$.  we define the convex order between two distribution functions in $\cal M$. For any distribution functions $F$ and $G$, $F \leq_{\rm cx} G$ refers to $\E[\phi(F)] \leq \E[\phi(G)]$ for all convex functions $\phi: \R \to \R$ provided that the expectations $\E[\phi(F)]$ and $\E[\phi(G)]$ exist. 	 For more details about stochastic order, we refer to \cite{SS07}.
	
If $F \in {\cal M}$ is IFR, then $\Fbar$ can only up-cross an exponential survival function (see Proposition \ref{pro-2-1}(iii)). Thus, $\delta_1 \leq_{\rm cx} F \leq_{\rm cx} {\rm Exp}(1)$ for $F\in {\cal F}$ (e.g., see Theorem 3.A.44 in \cite{SS07}). Due to the $\leq_{\rm  cx}$-consistent\footnote{A risk measure $\rho$ is $\leq_{\rm cx}$-consistent if $\rho(F) \leq \rho(G)$ whenever $F \leq_{\rm cx} G$.} of coherent distortion risk measure $\rho_h$, then $\rho_h(F) \in [1, \rho_h(E)]$ for any $F \in {\cal F}$, where $E$ is the exponential distribution function with parameter $1$. Moreover, the lower and upper bounds are attained by $\delta_1$ and ${\rm Exp}(1)$ distributions, respectively. We reformulate it as the following theorem.
	
\begin{theorem}  \label{th-coherent}
For any coherent distortion risk measure $\rho_h$, we have
\begin{equation*}
	\sup_{F \in {\cal F}} \rho_h (F) = \rho_h(E),\ {\rm and}\ \inf_{F \in {\cal F}} \rho_h (F) = 1,
\end{equation*}
where $E(x) = 1 - \exp\{-x\}$ for $x\in\R_+$.
\end{theorem}

\subsection{RVaR}\label{subsec33}

In this subsection, we consider the case of RVaR, i.e., \eqref{RVaR}. We only study the case for $0 \leq \alpha < \beta < 1$ by Theorems \ref{th-VaR} and \ref{th-coherent}. Unlike the straightforward proofs in previous subsections, the proofs in this subsection and Section \ref{sec4} are more technically involved and complex. Denote
\begin{align}  \label{eq-20250622}
       {\cal F}^\ast &= \left\{G_{t,w} \in {\cal F}:\ \Gbar_{t,w}(x) = \exp\{-wx\} \id_{\{x<t\}},\ w \in [0, 1],\ t \in [1, \oo] \right\},
\end{align}
where $w \in [0, 1]$ and $t \in [1, \oo]$ are chosen such that the mean of $G_{t,w}$ equals to $1$, i.e., $\mu(G_{t,w}) = 1$.
	
We first present Theorem \ref{thm3-4} for the worst-case RVaR, followed by a proposition facilitating the  worst-case scenario calculation, and finally Theorem \ref{thm3-7} for the best-case RVaR.
	
\begin{theorem}\label{thm3-4}
For $0 \leq \alpha < \beta < 1$, we have
\begin{equation}\label{eq-2024-38}
     \sup_{F \in {\cal F}} {\rm RVaR}_{\alpha, \beta}(F) = \sup_{F \in {\cal F}^\ast} {\rm RVaR}_{\alpha, \beta}(F).
\end{equation}
\end{theorem}
	
\begin{proof}
It is easy to see that $\sup_{F \in {\cal F}} {\rm RVaR}_{\alpha, \beta}(F) \geq \sup_{F \in {\cal F}^\ast} {\rm RVaR}_{\alpha, \beta}(F)$ since ${\cal F}^\ast \subseteq {\cal F}$. Below, we need to show the reverse inequality. It suffices to show that for any given $F \in {\cal F}$, there exists another distribution function $G_{t,w} \in {\cal F}^\ast$ such that
\begin{equation}\label{eq-2024-39}
	{\rm RVaR}_{\alpha, \beta}(F) \leq {\rm RVaR}_{\alpha, \beta}(G_{t,w}).
\end{equation}
In fact, \eqref{eq-2024-39} implies
$$
	{\rm RVaR}_{\alpha, \beta}(F) \leq \sup_{F \in {\cal F}^\ast} {\rm RVaR}_{\alpha, \beta}(F),
$$
which yields $\sup_{F \in {\cal F}} {\rm RVaR}_{\alpha, \beta}(F) \leq \sup_{F \in {\cal F}^\ast} {\rm RVaR}_{\alpha, \beta}(F)$. Thus, \eqref{eq-2024-38} holds.
		
In the following, we show how to construct a distribution function $G_{t,w} \in {\cal F}^\ast$ such that \eqref{eq-2024-39} holds. For $F \in {\cal F}$, $\Lambda_F(x)$ is a positive, continuous, and increasing convex function in $x \in [0, d)$ with at most one discontinuous point $d$, where $d := \sup\left\{t \geq 0: \Fbar(t) > 0\right\}$ by Proposition \ref{pro-2-1}. Denote by $x_p=\VaR_p(F)$ for $p\in [0,1]$. Below, we consider two cases:
		
\underline{\emph{Case 1:}}\ If $\Lambda_F(x)$ is discontinuous at $x_\alpha$, we have ${\rm RVaR}_{\alpha, \beta}(F) = x_\alpha$. Denote $p_0 := F\(x_\alpha-\) \leq \alpha$. Two subcases arise.
\begin{itemize}
  \item Subcase 1.1: Suppose that $-\ln(1-p_0) > x_\alpha$. In this subcase, we have $p_0 > 1 - \exp\{-x_\alpha\}$. Choose the ${\rm Exp}(1)$ distribution $E=G_{\oo,1} \in {\cal F}^\ast$. Then, $E\(x_\alpha\) < p_0$. For any $\gamma \in (\alpha, \beta]$, we get that
      \begin{equation*}
		\VaR_\gamma(E) > \VaR_{p_0}(E) \geq x_\alpha.
	  \end{equation*}
	  Therefore, ${\rm RVaR}_{\alpha, \beta}(E) \geq x_\alpha= {\rm RVaR}_{\alpha, \beta}(F)$.
			
 \item Subcase 1.2: Suppose that $-\ln(1-p_0)\leq x_\alpha$, i.e., $w:=-x_\alpha^{-1} \ln (1-p_0)\in (0,1]$; see Figure \ref{Fig3-1}. If $F \in {\cal F}^\ast$, then choose $G = F \in {\cal F}^\ast$. Thus, ${\rm RVaR}_{\alpha, \beta}(G) = {\rm RVaR}_{\alpha, \beta}(F) = x_\alpha$. If $F \notin {\cal F}^\ast$, then, by the convexity of $\Lambda_F$, we have $ \Lambda_F(x) < w x$ for all $x\in (0, x_\alpha)$. Hence, $\Fbar(x) > \exp\left\{ - w x\right\}$ for all $x \in (0, x_\alpha(F))$. Note that
	\begin{equation}\label{constraint1}
		1 = \mu(F) = \int_0^{x_\alpha} \Fbar(x) \d x > \int_0^{x_\alpha} \exp\{- w x\} \d x.
	\end{equation}
	Since $w\in (0,1]$, we have
	$$		\int_0^\oo\exp\{- w x\} \d x \geq \int_0^\oo \exp\{- x\} \d x = 1. 			$$
    In view of \eqref{constraint1}, there exists some $t \in (x_\alpha, \oo]$ such that $G_{t,w}\in {\cal F}^\ast$, where $G_{t,w}$ is defined in \eqref{eq-20250622}. Below, we show that ${\rm RVaR}_{\alpha, \beta}(G_{t,w}) > x_\alpha = {\rm RVaR}_{\alpha, \beta}(F)$. It is obvious that $\Lambda_{G_{t,w}}$ is continuous at $x_\alpha$. Also, for $\gamma \in (\alpha, \beta]$, there exists $\eta_\gamma > 0$ such that
    $$
	  \Lambda_{G_{t,w}}(x) =w x <-\ln(1-\gamma),\ \ \hbox{for}\ x\in (x_\alpha, x_\alpha+\eta_\gamma],
	$$
    which implies $G_{t,w}(x) < \gamma$ for $x \in (x_\alpha, x_\alpha+\eta_\gamma]$. Therefore, $\VaR_\gamma(G_{t,w}) > x_\alpha+ \eta_\gamma > x_\alpha$ for $\gamma\in (\alpha, \beta]$, so we get ${\rm RVaR}_{\alpha, \beta}(G_{t,w}) > x_\alpha = {\rm RVaR}_{\alpha, \beta}(F)$.
\end{itemize}
		
  \begin{figure}[htbp]
	\centering
	\begin{minipage}[b]{0.45\textwidth}
		\begin{tikzpicture}
			\tikzset{
				box/.style ={
					circle, 
					minimum width =1.5pt, 
					minimum height =1.5pt, 
					inner sep=1.5pt, 
					draw=black, 
					fill=white
					}
	     		}
			\draw[->,thick] (0,0) -- (7,0) node[right] {$x$};
			\draw[->,thick] (0,0) -- (0,5) node[above] { };			
			\draw[domain=0:4.34,color=black] plot (\x,{0.16*\x*\x});
			\node at (3.2,1){$\Lambda_F$};
			\draw[domain=0:6.1,color=red] plot (\x,{0.7*\x})node[right] {$\Lambda_{G_{t,w}}$};
			\node[box] at (4.3,3){};
			\node[box] at (6.1,4.25){};
			\draw  [dash pattern={on 1.5pt off 2pt}] (4.23,3) -- (0,3);
			\draw  [dash pattern={on 1.5pt off 2pt}] (6.08,4.2) -- (6.1,0);
			\draw  [dash pattern={on 1.5pt off 2pt}] (4.3,2.8) -- (4.3,0);
			\node at (4.3,-0.3) {\footnotesize $\VaR_{\alpha}(F)$};
			\node at (4.3,-0.6) {\footnotesize $(=x_\alpha)$};
			\node at (6.1,-0.3) {\footnotesize $t$};
			\node at (0,-0.3) {\footnotesize $0$};
			\node at (-0.3,3) {\footnotesize \rotatebox{90}{$-\ln(1\!-\! p_0)$}};
		\end{tikzpicture}
		\caption{The diagrams of $\Lambda_F$ and $\Lambda_{G_{t,w}}$.}
			\label{Fig3-1}
	\end{minipage}
	\hfill
	\begin{minipage}[b]{0.45\textwidth}
		\begin{tikzpicture}
			\tikzset{
				box/.style ={
					circle, 
					minimum width =1.5pt, 
					minimum height =1.5pt, 
					inner sep=1.5pt,  
					draw=black,  
					fill=white
				}
			}
			\draw[->,thick] (0,0) -- (7.5,0) node[right] {$x$};
     		\draw[->,thick] (0,0) -- (0,5) node[above] { };			
			\draw[domain=0:7,color=black] plot (\x,{0.1*\x*\x})node[right]{$\Lambda_F$};
			\draw[domain=0:6,color=red] plot (\x,{0.4*\x});
			\node at (6,2.70) {\textcolor{red}{$\Lambda_{G_{t,w}}$}};
			\node at (5.9,-0.3) {\footnotesize $t$};
			\node at (0,-0.3) {\footnotesize $0$};
			\node at (3.95,1.55){$\bullet$};
			\node[box] at (5.95,2.38){};
			\draw  [dash pattern={on 1.5pt off 2pt}] (3.9,1.58) -- (0,1.58);
			\draw  [dash pattern={on 1.5pt off 2pt}] (3.94,1.53) -- (3.94,0);
			\draw  [dash pattern={on 1.5pt off 2pt}] (5.94,2.33) -- (5.94,0);
			\node at (3.9,-0.3) {\footnotesize $\VaR_{\alpha}(F)$};
			\node at (3.9,-0.6) {\footnotesize $(=x_\alpha)$};
			\node at (-0.3,1.58) {\footnotesize \rotatebox{90}{$-\ln(1\!-\!\alpha)$}};
			\draw  [dash pattern={on 1.5pt off 2pt}] (6.5,4.225) -- (0,4.225);
			\draw  [dash pattern={on 1.5pt off 2pt}] (6.5,4.225) -- (6.5,0);
			\node at (6.75,-0.3) {\footnotesize $\VaR_{\beta}(F)$};
			\node at (-0.3,4.225) {\footnotesize \rotatebox{90}{$-\ln(1\!-\!\beta)$}};
		\end{tikzpicture}
		\caption{The diagrams of $\Lambda_F$ and $\Lambda_{G_{t,w}}$.}
		\label{Fig3-2}
  \end{minipage}
 \end{figure}
		
\underline{\emph{Case 2:}}\ If $\Lambda_F(x)$ is continuous at $x_\alpha$, we have $p_0 := F\(x_\alpha-\) = \alpha$. Next, we need to consider two subcases:
\begin{itemize}
  \item Subcase 2.1: Suppose that $-\ln(1-\alpha) \geq x_\alpha$. In this subcase, we choose the ${\rm Exp}(1)$ distribution $E \in {\cal F}^\ast$. By the convexity of $\Lambda_F$, we have
	$$
          \Lambda_F(x) \geq - \frac{\ln(1-\alpha)}{x_\alpha} x \geq x = \Lambda_E(x),\quad \forall x \in \[x_\alpha, x_\beta\].
	$$
    Then $x_p \leq \VaR_p(E)$ for $p \in [\alpha, \beta]$. Therefore, we have ${\rm RVaR}_{\alpha, \beta}(E) \geq {\rm RVaR}_{\alpha, \beta}(F)$.
			
\item Subcase 2.2: Suppose that $-\ln(1-\alpha) < x_\alpha$, i.e., $w:=-x_\alpha^{-1} \ln (1-p_0)\in (0,1]$; see Figure \ref{Fig3-2}. In this subcase, we have $\int_0^\oo \exp\{-w x\} \d x > 1$. 			Note that
	\begin{equation}\label{constraint2}
			1 = \mu(F) > \int_0^{x_\alpha} \Fbar(x) \d x \geq \int_0^{x_\alpha} \exp\{- w x \} \d x.
	\end{equation}
    In view of \eqref{constraint2}, there exists some $t\in (x_\alpha, \oo)$ such that $\int_0^t \exp\{-w x\} \d x = 1$. Thus, $G_{t,w} \in {\cal F}^\ast$, where $G_{t,w}$ is defined in \eqref{eq-20250622}.
		
	\begin{enumerate}[{\rm (1)}]
       \item If $t \geq x_\beta$, by the convexity of $\Lambda_F$, we have $\Lambda_F(x) \ge w x = \Lambda_{G_{t,w}}(x)$ for all $x \in \[x_\alpha, x_\beta\]$, which implies $x_p \leq \VaR_p(G_{t,w})$ for $p \in [\alpha, \beta]$. Therefore, we have ${\rm RVaR}_{\alpha, \beta}(G_{t,w}) \geq {\rm RVaR}_{\alpha, \beta}(F)$.
				
       \item If $t < x_\beta$, then $q_0 := F(t-) < \beta$. For any $q > \beta$, we have $\VaR_q(G_{t,w}) = t < x_\beta < x_q$. Therefore,
		   \begin{equation}\label{mean-right}
					\int_\beta^1 \VaR_q(G_{t,w}) \d q < \int_\beta^1 x_q \d q.
		   \end{equation}
          Moreover, by the convexity of $\Lambda_F$, we have $\Lambda_F(x) \leq w x = \Lambda_{G_{t,w}}(x)$ for all $x \in \[0, x_\alpha\)$, implying $\VaR_q(G_{t,w}) \leq x_q$ for any $q < \alpha$. Then
		  \begin{equation}\label{mean-left}
				\int_0^\alpha \VaR_q(G_{t,w}) \d q \leq \int_0^\alpha x_q \d q.
		  \end{equation}
         Since $\mu(F) = \mu(G_{t,w})=1$, it follows from \eqref{mean-right} and \eqref{mean-left} that
		 \begin{equation*}
				\int_\alpha^\beta \VaR_q(G_{t,w}) \d q \geq \int_\alpha^\beta x_q(F) \d q.
		\end{equation*}
		That is, ${\rm RVaR}_{\alpha, \beta}(G_{t,w}) \geq {\rm RVaR}_{\alpha, \beta}(F)$.
	\end{enumerate}
\end{itemize}
This completes the proof of the theorem.
\end{proof}
	
Theorem \ref{thm3-4} shows that an infinite dimensional optimization problem can be transformed into a finite dimensional optimization problem through the method of construction. Subsequently, we can use simulation software to perform numerical simulations on it. This construction method will run through this article. The next proposition \ref{pro3-5} provides a simple expression of the worst-case RVaR on uncertainty set ${\cal F}^\ast$.
	
\begin{proposition}  \label{pro3-5}
For $0 \leq \alpha < \beta \leq 1$, we have
\begin{equation}  \label{eq-250627}
     \sup_{F \in {\cal F}} {\rm RVaR}_{\alpha, \beta}(F) = \max_{w \in [\alpha, \beta]} \frac{(1-\beta)\ln(1-w) - (1-\alpha)\ln(1-\alpha) + w - \alpha}{w (\beta - \alpha)}.
\end{equation}
In particular, for $\alpha\in [0,1]$,
\begin{equation}   \label{eq-250628}
     \sup_{F \in {\cal F}} {\rm TVaR}_{\alpha}(F) = 1-\ln (1-\alpha).
\end{equation}
\end{proposition}
	
\begin{proof}
For any $G_{t,w} \in {\cal F}_w$, we have $1 - w = \exp\{-wt\}$. Thus, $t:=t_w = -w^{-1} \ln(1-w)$, where $t_0 = 1$ is the limit of $t_w$ when $w$ goes to $0$. It is easy to see that
$$
	\VaR_p(G_{t,w}) = \left\{\begin{array}{ll}
			-\dfrac {1}{w} \ln(1-p), & {\rm for}\ 0 \leq p \leq w, \\[8pt]
			-\dfrac{1}{w} \ln(1-w), & {\rm for}\ w < p \leq 1. \end{array} \right.
$$
Then
\begin{align*}
  g(w) := {\rm RVaR}_{\alpha, \beta}(G_{t,w}) &= \left\{\begin{array}{ll}
		-\dfrac{1}{w} \ln(1-w), & {\rm for}\ w \in [0, \alpha), \\[8pt]
		\dfrac{1}{\beta -\alpha} \[\int_\alpha^w \frac{-\ln(1-p)}{w}\d p -\frac {\beta - w}{w} \ln(1-w)\],
				& {\rm for}\ w \in [\alpha, \beta], \\[8pt]
		\dfrac{1}{\beta - \alpha} \int_\alpha^\beta \frac{-\ln(1-p)}{w} \d p, & {\rm for}\ w \in (\beta, 1],
			\end{array} \right.   \\[8pt]
	&= \left\{\begin{array}{ll}
		-\dfrac{1}{w} \ln(1-w), & {\rm for}\ w \in [0, \alpha), \\[8pt]
		\dfrac{(1-\beta)\ln(1-w) - (1-\alpha)\ln(1-\alpha) + w - \alpha}{w (\beta - \alpha)},
				& {\rm for}\ w \in [\alpha, \beta], \\[8pt]
		\dfrac{(1-\beta)\ln(1-\beta) - (1-\alpha)\ln(1-\alpha) + \beta - \alpha}{w (\beta - \alpha)},
			& {\rm for}\ w \in (\beta, 1]. \end{array} \right.
\end{align*}
It is easy to verify that $g(w)$ is increasing in $w \in [0, \alpha)$ and decreasing in $w \in (\beta, 1]$. Moreover, $g(w)$ is continuous in $w \in [0, 1]$. Therefore, we only consider $g(w)$ on $[\alpha, \beta]$. We first calculate its first-order derivative:
$$
     g'(w) = \frac{1}{w^2(1-w)(\beta-\alpha)} \[(1-w) \(\alpha + (1-\alpha) \ln(1-\alpha)\) - (1-\beta) \(w + (1-w)\ln(1-w)\)\].
$$
Denote $k(w) := (1-w) \(\alpha + (1-\alpha) \ln(1-\alpha)\) - (1-\beta) \(w + (1-w)\ln(1-w)\)$. Then $k'(w) = (1-\beta)\ln(1-w) - \(\alpha + (1-\alpha) \ln(1-\alpha)\) < 0$. Thus, $k(w)$ is non-increasing on $[\alpha, \beta]$. Combining $g'(\alpha) > 0$ and $g'(\beta) < 0$, we know that there exists $w^* \in [\alpha, \beta]$ such that $g'(w) \geq 0$ for $w \in [\alpha, w^*]$ and $g'(w) \leq 0$ for $w \in [w^*, \beta]$. This means that $g(w)$ achieves its maximum value at $w^* \in [\alpha, \beta]$. Therefore, \eqref{eq-250627} follows from Theorem \ref{thm3-4}.
It is trivial to check that \eqref{eq-250628} follows from \eqref{eq-250627} with $\beta=1$ by observing $\alpha+(1-\alpha)\ln (1-\alpha)\geq 0$ for $\alpha\in [0,1]$. This completes the proof of the proposition.
\end{proof}

\begin{remark}{\rm
\cite{CHLW23} recently derived the tight bounds for RVaR under mean or variance ambiguity set with support information. They showed that
\begin{equation*}
	\sup_{F \in \widetilde{{\cal F}}} {\rm RVaR}_{\alpha, \beta}(F) = \frac{1}{1-\alpha},
\end{equation*}
where $\widetilde{{\cal F}} := \left\{F: \mu(F) = 1 \right\}$. Since ${\cal F} \subseteq \widetilde{{\cal F}}$, it follows from Theorem \ref{thm3-4} that
$$
  	\sup_{F \in {\cal F}} {\rm RVaR}_{\alpha, \beta}(F) \leq \frac {1}{1-\alpha},\quad 0\le \alpha<\beta\le 1.
$$		}
\end{remark}
	
The next theorem characterizes the best-case RVaR on $\cal F$. Notably, it is uniquely attained by exponential or degenerate distributions depending on the parameters $\alpha$ and $\beta$.
	
\begin{theorem}\label{thm3-7}
For $0 \leq \alpha < \beta < 1$, we have
\begin{align}
	\inf_{F \in {\cal F}} {\rm RVaR}_{\alpha, \beta}(F)
		&= \min \{{\rm RVaR}_{\alpha, \beta}(E), \RVaR_{\alpha, \beta}(\delta_1)\} \label{eq-2024-315}  \\
		&= \min\left\{1+\frac{(1-\beta)\log(1-\beta) -(1-\alpha)\log(1-\alpha)}{\beta-\alpha}, 1\right\},
 		\label{eq-2024-316}
\end{align}
where $E(x) = 1-\exp\{-x\}$ for $x\ge 0$.
\end{theorem}
	
\begin{proof}
Eq. \eqref{eq-2024-316} can be calculated directly. Below, we only prove \eqref{eq-2024-315}. Note that $\inf_{F \in {\cal F}} {\rm RVaR}_{\alpha, \beta}(F) \leq \min \{{\rm RVaR}_{\alpha, \beta}(E), \RVaR_{\alpha, \beta}(\delta_1)\}$. Below, we need to show the reverse inequality. It suffices to show that for any given $F \in {\cal F}$ with $F \neq E$ and $F \neq \delta_1$, we have
\begin{equation}\label{eq-2024-317}
      \min \{{\rm RVaR}_{\alpha, \beta}(E),\RVaR_{\alpha, \beta}(\delta_1)\}\leq {\rm RVaR}_{\alpha, \beta}(F).
\end{equation}
Denote by $x_p=\VaR_p(F)$ for $p\in [0,1]$. Below, we consider two cases.
		
\underline{\it Case 1}:\ Suppose $-\ln(1-\beta) \leq x_\beta$. In this case, whether $F$ is continuous at $x_\beta$ (see Figure \ref{Fig3-3}) or discontinuous at $x_\beta$ (see Figure  \ref{Fig3-4}), we always have $F(x) \leq E(x)$ for any $x \in [0, x_\beta]$ by the convexity of $\Lambda_F$ and $E(x_\beta) \geq \beta$. Thus, $x_p \geq \VaR_p(E)$ for $p \in [\alpha, \beta]$. Finally,
\begin{equation*}
	\RVaR_{\alpha, \beta}(F) \geq \RVaR_{\alpha, \beta}(E) \geq \min \{{\rm RVaR}_{\alpha, \beta}(E),
		\RVaR_{\alpha, \beta}(\delta_1)\}.
\end{equation*}
Therefore, \eqref{eq-2024-317} holds.
		
 \begin{figure}[htbp]
	\centering
		\begin{minipage}[b]{0.45\textwidth}
			\begin{tikzpicture}
		    	\tikzset{
					box/.style ={
					circle, 
			 		minimum width =1.5pt, 
					minimum height =1.5pt, 
					inner sep=1.5pt, 
					draw=black, 
					fill=white
					}
				}
				\draw[->,thick] (0,0) -- (6.8,0) node[right] {$x$};
				\draw[->,thick] (0,0) -- (0,6) node[above] { };			
				\draw[domain=0:6.5,color=black] plot (\x,{0.1*\x*\x});
				\node at (6.4,3.6){$\Lambda_F$};
				\node at (0,-0.3) {\footnotesize $0$};
				\draw[domain=0:5.4,color=blue] plot (\x,\x)node[right] {$\Lambda_E$};
				\draw  [dash pattern={on 1.5pt off 2pt}] (5.96,3.55) -- (5.96,0);
				\draw  [dash pattern={on 1.5pt off 2pt}] (5.96,3.55) -- (0,3.55);
				\draw  [dash pattern={on 1.5pt off 2pt}] (3.5359,3.55) -- (3.5359,0);
				\node at (5.96,-0.3) {\footnotesize $\VaR_{\beta}(F)$};
				\node at (3.5359,-0.3) {\footnotesize $\VaR_{\beta}(E)$};
				\node at (-0.3,3.55) {\footnotesize \rotatebox{90}{$-\ln(1-\beta)$}};
			\end{tikzpicture}
			\caption{The diagrams of $\Lambda_F$ and $\Lambda_E$.}
			\label{Fig3-3}
		\end{minipage}
		\hfill
		\begin{minipage}[b]{0.45\textwidth}
			\begin{tikzpicture}
				\tikzset{
					box/.style ={
						circle, 
						minimum width =1.5pt, 
						minimum height =1.5pt, 
						inner sep=1.5pt, 
						draw=black,  
						fill=white
					}
				}
				\draw[->,thick] (0,0) -- (6.8,0) node[right] {$x$};
				\draw[->,thick] (0,0) -- (0,6) node[above] { };			
				\draw[domain=0:5.95,color=black] plot (\x,{0.1*\x*\x});
				\node at (6.4,3.6){$\Lambda_F$};
				\node[box] at (5.94,3.57){};
				\draw[domain=0:5.4,color=blue] plot (\x,\x)node[right] {$\Lambda_E$};
				\draw  [dash pattern={on 1.5pt off 2pt}] (5.96,3.55) -- (5.96,0);
				\draw  [dash pattern={on 1.5pt off 2pt}] (3.5359,3.55) -- (3.5359,0);
				\draw  [dash pattern={on 1.5pt off 2pt}] (5.96,3.55) -- (0,3.55);
				\node at (3.5359,-0.3) {\footnotesize $\VaR_{\beta}(E)$};
				\node at (5.96,-0.3) {\footnotesize $\VaR_{\beta}(F)$};
				\node at (0,-0.3) {\footnotesize $0$};
				\node at (-0.3,3.55) {\footnotesize \rotatebox{90}{$-\ln(1-\beta)$}};
			\end{tikzpicture}
			\caption{The diagrams of $\Lambda_F$ and $\Lambda_E$.}
			\label{Fig3-4}
		\end{minipage}
 \end{figure}
		
\underline{\it Case 2}: \ Suppose $-\ln(1-\beta) > x_\beta$. In this case, we first show $x_\beta> 1$. If $x_\beta \leq 1$, then $F$ will be continuous at $x_\beta$; otherwise, $\esssup (F)\le x_\beta$. Hence,
$$
	\mu(F) = \int_0^1 x_p \d p=  \int_0^\beta x_p \d p + (1-\beta)< 1,
$$
which contradicts the condition $\mu(F) = 1$. By \eqref{inf-IFR}, we have $\beta = F(x_{\beta}) = F(x_{\beta}-) \le 1- \exp\left\{-x_{\beta}\right\} < \beta$. This also leads to a contradiction. Thus, $x_\beta> 1$. Therefore, there exists some point $t \in (0, 1)$ such that
$$   	\frac{x_\beta- t}{1-t} = -\ln(1-\beta).    $$
Next, we construct a distribution function $H_t \in {\cal F}$ with $\Lambda_{H_t}$:
$$
	\Lambda_{H_t}(x) = \left\{\begin{array}{ll} 0, & {\rm for}\  x \in [0, t], \\[3pt]
			\dfrac{x-t}{1-t}, & {\rm for}\  x > t. \end{array} \right.
$$
Then, $x_\beta = \VaR_\beta(H_t)$. By the convexity of $\Lambda_F$, two situations will arise, i.e., whether $F(x) \leq H_t(x)$ always holds for $x \in (t, x_\beta)$  or $F$ down-crosses $H_t$ at $t' \in (t, x_\beta)$. In the former case, $x_p \geq \VaR_p(H_t)$ for $p \in [\alpha, \beta]$, implying
$$
	\RVaR_{\alpha, \beta}(F) \geq \RVaR_{\alpha, \beta}(H_t).
$$
In the later case, denote by $\alpha' := F(t') = H_t(t')$.
\begin{itemize}
   \item If $\alpha \geq \alpha'$, it is easy to see that $\RVaR_{\alpha, \beta}(F) \geq \RVaR_{\alpha, \beta}(H_t)$.
			
   \item If $\alpha < \alpha'$, we have $H_t(x) \leq F(x)$ for $x \in [0, t']$. Then $x_p \leq \VaR_p(H_t)$ for $p \in [0, \alpha]$ and $\TVaR_\alpha(H_t) \leq \TVaR_\alpha(F)$ since $\mu(F) = \mu(H_t) = 1$.
	  \begin{itemize}
        \item Assume that $F$ is not continuous at $x_\beta$; see Figure \ref{Fig3-5}. We obtain $\VaR_p(H_t) \geq \VaR_\beta(H_t) = x_\beta =x_p$ for $p \in [\beta, 1]$. Thus, $\TVaR_\beta(F) \leq \TVaR_\beta(H_t)$ holds.
				
        \item Assume that $F$ is continuous at $x_\beta$; see Figure \ref{Fig3-6}. We obtain $F(x) \geq H_t(x)$ for $x \geq x_\beta$. Then, $x_p \leq \VaR_p(H_t)$ holds for $p \geq \beta$. Thus, $\TVaR_\beta(F) \leq \TVaR_\beta(H_t)$ holds.
	  \end{itemize}
	  Therefore, whether $F$ is discontinuous or continuous at $x_\beta$, we have
	 \begin{figure}[htbp]
			\centering
		\begin{minipage}[b]{0.45\textwidth}
			\begin{tikzpicture}
				\tikzset{
					box/.style ={
					circle,
					minimum width =1.5pt,
					minimum height =1.5pt,
					inner sep=1.5pt,
					draw=black,
					fill=white
		    		}
	    		}
			\draw[->,thick] (0,0) -- (6.6,0) node[right] {$x$};
			\draw[->,thick] (0,0) -- (0,6.5) node[above] { };			
			\draw[domain=0:5.28,color=black] plot (\x,{0.13*\x*\x});
			\node at (4.65,2.3){$\Lambda_F$};
			\draw[domain=1.5:5.8,color=blue] plot (\x,{1.3*(\x-1.5)})node[right] {$\Lambda_{H_t}$};
				\node[box] at (5.28,3.58){};
				\node at (1.8,0.4){$\bullet$};
				\draw  [dash pattern={on 1.5pt off 2pt}] (5.28,4.9) -- (0,4.9);
				\draw  [dash pattern={on 1.5pt off 2pt}] (5.28,4.9) -- (5.28,0);
				\draw  [dash pattern={on 1.5pt off 2pt}] (5.28,3.58) -- (0,3.58);
				\draw  [dash pattern={on 1.5pt off 2pt}] (1.8,0.4) -- (1.8,0);
				\draw  [dash pattern={on 1.5pt off 2pt}] (1.8,0.4) -- (0,0.4);
				\node at (5.28,-0.3) {\footnotesize $\VaR_{\beta}(F)$};
				\node at (-0.3,3.3) {\footnotesize \rotatebox{90}{$-\ln(1\!-\!p_0)$}};
				\node at (-0.3,5.3) {\footnotesize \rotatebox{90}{$-\ln(1\!-\!\beta)$}};
				\node at (1.85,-0.3) {\footnotesize $t'$};
				\node at (1.35,-0.3) {\footnotesize $t$};
				\node at (0,-0.3) {\footnotesize $0$};
				\node at (-0.3,0.643078) {\footnotesize \rotatebox{90}{$-\ln(1\!-\!\alpha')$}};
			\end{tikzpicture}
			\caption{The diagrams of $\Lambda_F$ and $\Lambda_{H_t}$.}
			\label{Fig3-5}
		\end{minipage}
		\hfill
		\begin{minipage}[b]{0.45\textwidth}
			\begin{tikzpicture}
				\tikzset{
					box/.style ={
					circle,
					minimum width =1.5pt,
					minimum height =1.5pt,
					inner sep=1.5pt,
					draw=black,
					fill=white
					}
				}
			\draw[->,thick] (0,0) -- (6.6,0) node[right] {$x$};
			\draw[->,thick] (0,0) -- (0,6.5) node[above] { };			
			\draw[domain=0:5.8,color=black] plot (\x,{0.15*(\x+0.6)*(\x+0.2)})node[left]{};
			\draw[domain=1.2:6.1,color=blue] plot (\x,{1.2*(\x-1.2)})node[below right]{$\Lambda_{H_t}$};
			\node at (5.45,5.1){$\bullet$};
			\node at (3.9,2.3) {$\Lambda_F$};
			\draw  [dash pattern={on 1.5pt off 2pt}] (5.45,5.1) -- (0,5.1);
			\draw  [dash pattern={on 1.5pt off 2pt}] (5.45,5.1) -- (5.45,0);
			\node at (5.45,-0.3) {\footnotesize $\VaR_{\beta}(F)$};
			\node at (-0.3,5.1) {\footnotesize \rotatebox{90}{$-\ln(1\!-\!\beta)$}};
			\node at (1.76,0.7){$\bullet$};
			\draw  [dash pattern={on 1.5pt off 2pt}] (1.76,0.7) -- (1.76,0);
			\draw  [dash pattern={on 1.5pt off 2pt}] (0,0.7) -- (1.76,0.7);
    		\node at (1.76,-0.3) {\footnotesize $t'$};
			\node at (1.2,-0.3) {\footnotesize $t$};
			\node at (0,-0.3) {\footnotesize $0$};
			\node at (-0.3,0.7) {\footnotesize \rotatebox{90}{$-\ln(1\!-\!\alpha')$}};
		\end{tikzpicture}
		\caption{The diagrams of $\Lambda_F$ and $\Lambda_{H_t}$.}
		\label{Fig3-6}
	\end{minipage}
  \end{figure}
  \begin{align*}
    \RVaR_{\alpha, \beta}(F) &=\frac{1}{\beta-\alpha}\[(1-\alpha)\TVaR_\alpha(F)-(1-\beta) \TVaR_\beta(F)\]\\
		&\geq \frac{1}{\beta -\alpha} \[(1-\alpha) \TVaR_\alpha(H_t) - (1-\beta) \TVaR_\beta(H_t)\]
				= \RVaR_{\alpha, \beta}(H_t).
  \end{align*}
\end{itemize}
Note that for $p \in (0, 1)$, $\VaR_p(H_t) = t - (1-t) \ln(1-p) = t \VaR_p(\delta_1) + (1-t) \VaR_p(E)$. 		Thus, we have
\begin{align*}
	\RVaR_{\alpha, \beta}(F) &\geq \RVaR_{\alpha, \beta}(H_t)
			= t\,\RVaR_{\alpha, \beta}(\delta_1) + (1-t)\,\RVaR_{\alpha, \beta}(E) \\
			&\geq \min \{{\rm RVaR}_{\alpha, \beta}(E), \RVaR_{\alpha, \beta}(\delta_1)\}.
\end{align*}
This completes the proof of the theorem.
\end{proof}

\section{Extreme-cases via mean and variance}\label{sec4}

In this section, we consider the mean-variance ambiguity set of IFR distribution $F \in {\cal M}$ with its dispersion information specified by the second moment of the distribution:
\begin{equation*}
		{\cal G} := \left\{F \in {\cal M}: F {\rm\ is\ IFR},\ \mu(F) = 1,\ \mu_2(F) = \mu_2\right\}.
\end{equation*}
Let's recall two facts \citep{BMP63}. The first one is that $F \in {\cal G}$ implies $1 \leq \mu_2(F) \leq 2$. The another one is that $F \in {\cal G}$ contains only one distribution when $\mu_2(F) = 1$ or $\mu_2(F) = 2$. Specifically, ${\cal G}$ degenerates to (i) the Dirac delta at $1$, $\delta_1$, when $\mu_2(F) = 1$, and (ii) ${\rm Exp}(1)$ distribution when $\mu_2(F) = 2$. Thus, from now on, we assume that $1 < \mu_2(F) < 2$ to avoid trivial discussions when $F \in {\cal G}$. We consider the following questions:
\begin{equation}   \label{eq-2024-1}
		\sup_{F \in {\cal G}} \rho(F), {\rm\ and}\ \inf_{F \in {\cal G}} \rho(F),
\end{equation}
where $\rho$ is a risk measure.
	
In order to solve \eqref{eq-2024-1}, we need to review some notations and lemmas \citep{BM64}. Denote $T_0 = 1 - \sqrt{\mu_2 - 1}$, and define
\begin{eqnarray}\label{eq-2024-42}
	\Gbar_{T_1}(x) &=& \left\{\begin{array}{ll}  1, & {\rm for}\ x < 0, \\
			\exp\{-ax\}, & {\rm for}\ 0 \leq x < T_1, \\
			0, & {\rm for}\ x \geq T_1, \end{array} \right.
\end{eqnarray}
where $a$ and $T_1$ are chosen such that $\mu(G_{T_1}) = 1$ and $\mu_2(G_{T_1}) = \mu_2$. The first two moment conditions uniquely identify $a \in [0, 1]$ and $T_1 \geq 1$. Note that $T_0 < T_1$. For $T \geq T_1$, define
\begin{eqnarray}   \label{eq-2024-43}
	\Gbar_T^{(1)}(x) &=& \left\{\begin{array}{ll} 	1, & {\rm for}\ x < \Delta, \\
			\exp\{-k(x-\Delta)\}, & {\rm for}\ \Delta \leq x < T, \\
			0, & {\rm for}\ x \geq T, \end{array} \right.
\end{eqnarray}
and for $T_0 \leq T < T_1$, define
\begin{eqnarray}   \label{eq-2024-44}
	\Gbar_T^{(2)}(x) &=& \left\{\begin{array}{ll}  	1, & {\rm for}\ x < 0, \\
			\exp\{-k_1 x\}, & {\rm for}\ 0 \leq x < T, \\
			\exp\{-k_1 T - k_2(x-T)\}, & {\rm for}\ x \geq T, \end{array} \right.
\end{eqnarray}
where $(k, \Delta)$ and $(k_1, k_2)$ are respectively determined by the moment conditions, i.e.,
\begin{equation}   \label{eq-2024-45}
	1=\int_0^\oo \Gbar^{(i)}_T(x) \d x\ \ {\rm and}\ \ \mu_2=2\int_0^\oo x \Gbar^{(i)}_T(x)\d x,\quad i=1, 2.
\end{equation}
According to \cite{BM64}, for any fixed $T$, Eq. \eqref{eq-2024-45} has uniquely determined solution pairs $(k(T), \Delta(T))$ and $(k_1(T), k_2(T))$, both of which are continuous in $T\in [T_0, \oo)$. Denote
\begin{equation}\label{distributions}
      {\cal G}_1 := \left\{G_T^{(1)}: T \geq T_1\right\}\ \ {\rm and}\ \ {\cal G}_2 := \left \{G_T^{(2)}: T_0 \leq T < T_1\right \}.
\end{equation}
We now state a useful lemma.
	
\begin{lemma}{\rm \citep[][Theorem 3.1]{BM64}}\label{lem-4-1}
For any given $1 < \mu_2 < 2$ and $x \geq 0$, we have
$$
   	\left\{(x, F(x)): F\in {\cal G}\right\} = \left\{(x, G(x)): G\in {\cal G}_1\cup {\cal G}_2\right\}.
$$
\end{lemma}
	
Our analysis commences with the derivation of extreme VaR cases, illustrating the power of Lemma \ref{lem-4-1}.

\begin{proposition}
For any $\alpha \in (0, 1)$, we have
\begin{equation}\label{eq-2024-6}
      \sup_{F \in {\cal G}} \VaR_\alpha(F) = \sup_{F \in {\cal G}_1 \cup {\cal G}_2} \VaR_\alpha(F),\qquad
			\inf_{F \in {\cal G}} \VaR_\alpha(F) = \inf_{F \in {\cal G}_1 \cup {\cal G}_2} \VaR_\alpha(F),
\end{equation}
and
\begin{equation}\label{eq-2024-7}
	\sup_{F \in {\cal G}} \VaR_\alpha^+(F) = \sup_{F \in {\cal G}_1 \cup {\cal G}_2} \VaR_\alpha^+(F),\qquad
			\inf_{F \in {\cal G}} \VaR_\alpha^+(F) = \inf_{F \in {\cal G}_1 \cup {\cal G}_2} \VaR_\alpha^+(F).
\end{equation}
\end{proposition}
	
\begin{proof}
We give the proof of \eqref{eq-2024-6} only; the proof of \eqref{eq-2024-7} is similar. Since ${\cal G}_1 \subseteq {\cal G}$ and ${\cal G}_2 \subseteq {\cal G}$, it follows that
$$
     \sup_{F \in {\cal G}} \VaR_\alpha(F) \geq \sup_{F \in {\cal G}_1 \cup {\cal G}_2} \VaR_\alpha(F)\ \ {\rm and}\ \ \inf_{F\in {\cal G}} \VaR_\alpha(F)\leq \inf_{F\in {\cal G}_1\cup {\cal G}_2} \VaR_\alpha(F).
$$
For the reverse inequality of \eqref{eq-2024-6}, note that for any given $F \in {\cal G}$, we have $F(\VaR_\alpha(F) - \ep) < \alpha$ for any $\ep > 0$ and $F(\VaR_\alpha(F)) \geq \alpha$. By Lemma \ref{lem-4-1}, there exist two distribution functions $G_1, G_2 \in {\cal G}_1 \cup {\cal G}_2$ such that $G_1(\VaR_\alpha(F) - \ep) = F(\VaR_\alpha(F) - \ep) < \alpha$ and $G_2(\VaR_\alpha(F)) \geq \alpha$. Then $\VaR_\alpha(G_1) > \VaR_\alpha(F) - \ep$ and $\VaR_\alpha(G_2) \leq \VaR_\alpha(F)$. Thus,
$$
	\sup_{G \in {\cal G}_1 \cup {\cal G}_2} \VaR_\alpha(G) \geq \VaR_\alpha(G_1) > \VaR_\alpha(F) - \ep,
$$
and
$$
	\inf_{G\in {\cal G}_1 \cup {\cal G}_2} \VaR_\alpha(F) \leq \VaR_\alpha(G_2) \leq \VaR_\alpha(F).
$$
Therefore, the desired result is obtained by taking the limit as $\ep \downarrow 0$ and computing the supremum of $F$ over the set ${\cal G}$.
\end{proof}
	
We now examine the extreme-case RVaR. Building upon the cut criterion established in \cite{KN63} (which plays a pivotal role in our proof), we first establish the simple extreme-case TVaR result. This preliminary result will serve to clarify our subsequent RVaR analysis.

\begin{lemma}\label{cut}
For any two distribution functions $F, G \in {\cal M}$ with $\mu(F) = \mu(G)$, if $F$ up-crosses $G$ with only one time, then $\mu_2(F) < \mu_2(G)$.
\end{lemma}

\begin{theorem}\label{thm-4-4}
For any $\alpha \in (0,1)$, we have
\begin{equation}\label{eq-2024-8}
	\sup_{F \in {\cal G}} \TVaR_\alpha(F) = \sup_{F \in {\cal G}_1} \TVaR_\alpha(F).
\end{equation}
\end{theorem}
	
\begin{proof}
Since $\sup_{F \in {\cal G}} \TVaR_\alpha(F) \geq \sup_{F \in {\cal G}_1} \TVaR_\alpha(F)$, it suffices to prove that for any $\alpha \in (0, 1)$ and given $F \in {\cal G}$, there exists another distribution function $G \in {\cal G}_1$ such that
\begin{equation}\label{eq-2024-9}
	\TVaR_\alpha(F) \leq \TVaR_\alpha(G).
\end{equation}
		
For any given $F \in {\cal G}$, let's first examine the relationship between $F$ and $G_{T_1}$, where $G_{T_1}$ is given by \eqref{eq-2024-42}. If $F = G_{T_1}$, then \eqref{eq-2024-9} holds since $G_{T_1} \in {\cal G}_1$. If $F \neq G_{T_1}$, note that $F \in {\cal G}$ implies that $\Lambda_F(x)$ is a convex function and $\Lambda_{G_{T_1}}(x)$ is a linear function in $x \in (0, T_1)$. Thus, by Proposition \eqref{pro-2-1} (iii), there exists some point $t_1 \in (0, T_1)$ such that $F$ up-crosses $G_{T_1}$ at $t_1$ under the conditions of the first two moments are equal and $F(0) = G_{T_1}(0) = 0$. Denote $F(t_1) = G_{T_1}(t_1) =: \alpha_1$. Let $G^{(1)}_\oo \in {\cal G}_1$ be defined by
\begin{eqnarray}\label{eq-2024-10}
	\Gbar^{(1)}_\oo(x) &=& \left\{\begin{array}{ll}  1, & {\rm for}\ x < \Delta_0, \\
		\exp\{-b(x-\Delta_0)\}, & {\rm for}\ x \geq \Delta_0, \end{array} \right.
\end{eqnarray}
where $b > 0$ such that $\mu(G^{(1)}_\oo) = 1$ and $\mu_2(G^{(1)}_\oo) = \mu_2(F)$. Now, we consider the link between $F \in {\cal G}$ and $G^{(1)}_\oo \in {\cal G}_1$.  If $F = G^{(1)}_\oo$, then \eqref{eq-2024-9} holds since $G^{(1)}_\oo \in {\cal G}_1$. If $F \neq G^{(1)}_\oo$ and $F(\Delta_0) = 0 = G_\oo(\Delta_0)$, it follows from Proposition \eqref{pro-2-1} (iii) that $F$ up-crosses $G^{(1)}_\oo$ on $(\Delta_0, \oo)$ at most once. Thus, the condition of the first two moments being equal doesn't hold by Lemma \ref{cut}. Thus, $F(\Delta_0) > 0$. In this case, $F$ down-crosses first and then up-crosses $G^{(1)}_\oo$ on $(\Delta_0, \oo)$ under the condition that the first two moments are equal. Denote by $t_0 \in (\Delta_0, \oo)$ the point at which $F$ up-crosses $G^{(1)}_\oo$, and set $\alpha_0 := G^{(1)}_\oo(t_0)$. For the graphical explanation of $\alpha_0$ and $\alpha_1$, see Figure \ref{Fig1}, where $t_0=r(\oo)$ and $t_1=r(T_1)$. It is worth pointing out that the two numbers $\alpha_0$ and $\alpha_1$ play important roles in the remaining proof of the theorem.
		
Below, we show how to construct $G \in {\cal G}_1$ such that \eqref{eq-2024-9} holds for any given $F \in {\cal G}$ with $F \neq G_{T_1}$ and $F \neq G^{(1)}_\oo$. We consider three cases.
		
\underline{\emph{Case 1:}}\ Suppose $\alpha \geq \alpha_0$. In this case, we have $F(x) > G^{(1)}_\oo(x)$ for $x > t_0$ by the convexity of $\Lambda_F$, see Figure \ref{Fig1}. It is obvious to see that $\VaR_p(F) \leq \VaR_p\big (G^{(1)}_\oo\big )$ for any $p > \alpha_0$. Thus, we have
\begin{equation*}
	\TVaR_\alpha(F) \leq \TVaR_\alpha\big (G^{(1)}_\oo\big ).
\end{equation*}
		
\underline{\emph{Case 2:}}\ Suppose $\alpha \leq \alpha_1$. In this case, we have $F(x) < G_{T_1}(x)$ for any $x \in (0, t_1)$ by the convexity of $\Lambda_F$, see Figure \ref{Fig1}. Thus, $\VaR_p(F) \geq \VaR_p(G_{T_1})$ for any $p < \alpha_1$, which implies
\begin{equation*}
	\int_0^\alpha \VaR_p(F)\d p \geq \int_0^\alpha \VaR_p(G_{T_1}) \d p.
\end{equation*}
Therefore,
\begin{align*}
	\TVaR_\alpha(F) &= \frac{1}{1-\alpha} \(1 - \int_0^\alpha \VaR_p(F)\d p\) \\
		&\leq \frac{1}{1-\alpha} \(1 - \int_0^\alpha \VaR_p(G_{T_1}) \d p\) = \TVaR_\alpha (G_{T_1}).
\end{align*}
		
\underline{\emph{Case 3:}}\ Suppose $\alpha_0 \wedge \alpha_1 < \alpha < \alpha_0 \vee \alpha_1$. For any given $F \in {\cal G}$, we give some analysis on $F$ and $G^{(1)}_T$, where $G^{(1)}_T \in {\cal G}_1$ is given by \eqref{eq-2024-43} with $T \geq T_1$. If $F(\Delta) = 0 = G_T^{(1)} (\Delta)$, then $F$ up-crosses $G^{(1)}_T$ on $(\Delta, T)$ only once by Proposition \eqref{pro-2-1} (iii). If $F(\Delta) > 0 = G^{(1)}_T(\Delta)$, we know that $F$ down-crosses first and then up-crosses $G_T^{(1)}$ on $(\Delta, T)$ under the condition that the first two moments are equal. Thus, whether $F(\Delta) = 0$ or $F(\Delta) > 0$, $F$ up-crosses $G^{(1)}_T$ on $(\Delta, T)$ only once. Assume $F$ up-crosses $G^{(1)}_T$ at $r(T) \in (\Delta, T)$ and denote
$$
          g(T) := G_T^{(1)}(r(T)).
$$
The intersection point $(r(T), -\ln(1-g(T)))$ has been graphically represented in in Figure \ref{Fig2}. Note that $g(T_1) = G_{T_1}(r(T_1)) = G_{T_1}(t_1) = \alpha_1$ and $g(\oo) = G^{(1)}_\oo(r(\oo)) = G^{(1)}_\oo(t_0) = \alpha_0$. We will establish \eqref{eq-2024-9} by verifying the following two claims.
		
\begin{figure}[htbp]
	\centering
	\begin{minipage}[b]{0.45\textwidth}
		\begin{tikzpicture}
			\tikzset{
				box/.style ={
					circle,
					minimum width =1.5pt,
					minimum height =1.5pt,
					inner sep=1.5pt,
					draw=black,
					fill=white
				}
			}
			\draw[->,thick] (0,0) -- (6.5,0) node[right] {$x$};
			\draw[->,thick] (0,0) -- (0,6.5) node[above] { };			
			\draw[domain=0:5.3] plot (\x,{0.2*\x*\x})node[above] {$\Lambda_F$};
			\draw[domain=0:4,color=red] plot (\x,{0.6*\x})node[below right] {$\Lambda_{{G}_{T_1}}$};
			\draw[domain=1:5.5,color=blue] plot (\x,{1.2*(\x-1)})node[right] {$\Lambda_{G^{(1)}_{\oo}}$};
			\node at (3,1.77) { $\bullet$};
			\node at (4.3,1.6) {\footnotesize $(r(T_1), -\ln(1-\alpha_1))$};
			\node[box] at (3.95,2.35){};
			\node at(4.73205,4.47846){ $\bullet$};
			\node at (6.1,4.35) {\footnotesize $(r(\oo), -\ln(1-\alpha_0))$};
			\draw  [dash pattern={on 1.5pt off 2pt}] (3,1.7) -- (3,0);
            \draw  [dash pattern={on 1.5pt off 2pt}] (4.73,4.24) -- (4.73,0);
            \node at (3,-0.3) {\footnotesize $t_1$};
            \node at (4.73,-0.3) {\footnotesize $t_0$};
            \node at (0,-0.3) {\footnotesize $0$};
		\end{tikzpicture}
		\caption{The graphs of $\Lambda_F$, $\Lambda_{G^{(1)}_{\oo}}$ and $\Lambda_{{G}_{T_1}}$.}
		\label{Fig1}
	\end{minipage}
	\hfill
	\begin{minipage}[b]{0.45\textwidth}
		\begin{tikzpicture}
			\tikzset{
				box/.style ={
					circle,
					minimum width =1.5pt,
					minimum height =1.5pt,
					inner sep=1.5pt,
					draw=black,
					fill=white
					}
				}
			\draw[->,thick] (0,0) -- (6,0) node[right] {$x$};
			\draw[->,thick] (0,0) -- (0,6) node[above] { };			
			\draw[domain=0:5.2] plot (\x,{0.2*\x*\x})node[above] {$\Lambda_{F}$};
			\draw[domain=1:5,color=red] plot (\x,{1.1*(\x-1)})node[right] {$\Lambda_{G^{(1)}_T}$};
			\node at (4.2,3.5) { $\bullet$};
			\node at (5.5,3.2) {\footnotesize $(r(T), -\ln(1\!-\!g(T)))$};
			\node[box] at (5,4.4){};
			\node at (3.0,5.5) {\footnotesize $T>T_1$};
			\node at (1.32,0.35) { $\bullet$};
			\node at (1.32,-0.3) {\footnotesize $t'$};
            \node at (4.2,-0.3) {\footnotesize $t^\ast$};
            \draw  [dash pattern={on 1.5pt off 2pt}] (1.32,0.3) -- (1.32,0);
            \draw  [dash pattern={on 1.5pt off 2pt}] (4.2,3.1) -- (4.2,0);
            \node at (0,-0.3) {\footnotesize $0$};
			\draw  [dash pattern={on 1.5pt off 2pt}] (0.73,0) -- (0.73,0.12);
		\end{tikzpicture}
		\caption{The graphs of $\Lambda_F$ and $\Lambda_{G^{(1)}_T}$.}
			\label{Fig2}
	\end{minipage}
 \end{figure}
		
\begin{itemize}
  \item Claim 1: The function $g(T)$ is continuous on $T \in [T_1, \oo]$.
			
      \emph{Proof}: It is well known that $k(T) := k$ and $\Delta(T) := \Delta$ are continuous functions on $T \in [T_1, \oo]$, see Lemma 3.4 in \cite{BM64}. Below, we show that $k(T)$ and $\Delta(T)$ are increasing on $T \in [T_1, \oo]$. Taking $T_1 \leq T < T' \leq \oo$. If $\Delta(T) > \Delta(T')$, then $\Lambda_{G^{(1)}_T}$ and $\Lambda_{G^{(1)}_{T'}}$ have at most one intersection point. If there is no intersection point, then $\Lambda_{G^{(1)}_T}(x) \leq \Lambda_{G^{(1)}_{T'}}(x)$ on $x \in (0, \Delta(T))$ and $\Lambda_{G^{(1)}_T}(x) < \Lambda_{G^{(1)}_{T'}}(x)$ on $x \in (\Delta(T), T)$, which contracts $\mu(G^{(1)}_T) = \mu(G^{(1)}_{T'})$. If there is one intersection point between $\Lambda_{G^{(1)}_T}$ and $\Lambda_{G^{(1)}_{T'}}$, then this contracts $\mu_2(G^{(1)}_T) = \mu_2(G^{(1)}_{T'})$ by Lemma \ref{lem-4-1}. Thus, we have $\Delta(T) \leq \Delta(T')$. On the other hand, if $k(T) \geq k(T')$, then $G^{(1)}_T(x) \geq G^{(1)}_{T'}(x)$ on $x \in (0, T)$ and $G^{(1)}_T(x) > G^{(1)}_{T'}(x)$ on $x \in (T, T')$. This contracts $\mu(G^{(1)}_T) = \mu(G^{(1)}_{T'})$. Thus, $k(T) < k(T')$.
			
      Moreover, $r(T)$ is also continuous in $T \in [T_1, \oo]$ from the proof of Theorem 3.1 and Lemma 3.4 in \cite{BM64}. By the definition of $G^{(1)}_T \in {\cal G}_1$ given by \eqref{eq-2024-43}, $G^{(1)}_T(x)$ is continuous in $T \in [T_1, \oo]$ for any given fixed $x \geq 0$. Thus, taking any point $T' \in (T-\ep, T+\ep)$ for any $\ep > 0$ (if $T = T_1$, let $T' \in (T, T+\ep)$), we have
	  \begin{align*}
		\left|g(T)- g(T')\right | &= \left |G^{(1)}_T(r(T)) - G^{(1)}_{T'}(r(T'))\right| \\[2pt]
			&\leq \left|G^{(1)}_T(r(T)) - G^{(1)}_{T'}(r(T))\right| + \left|G^{(1)}_{T'}(r(T))
			- G^{(1)}_{T'}(r(T'))\right| \\[3pt]
    		&\leq \left|G^{(1)}_T(r(T)) - G^{(1)}_{T'}(r(T))\right |+k(T')\left |r(T)-r(T')\right| \\[3pt]
			&\leq  \left|G^{(1)}_T(r(T)) - G^{(1)}_{T'}(r(T))\right| + k(\oo) \left|r(T) - r(T')\right|,
	  \end{align*}
      where the second inequality follows from the Lipschitz continuity of $G_{T'}$ with Lipschitz constant $k(T')$, and the last one follows from the monotonicity $k(T)$ with $k(\oo) \in \R$. Letting $T' \to T$ yields that $\left|g(T) - g(T')\right| \to 0$. Thus, $g(T)$ is continuous in $T \in [T_1, \oo]$.
			
   \item Claim 2: For any $F \in {\cal G}$ and $\alpha_0 \wedge \alpha_1 < \alpha < \alpha_0 \vee \alpha_1$, we have
	   \begin{equation}\label{eq-2024-11}
		  \TVaR_\alpha(F) \leq \TVaR_\alpha \big (G^{(1)}_T\big ),
	   \end{equation}
      where $G^{(1)}_T \in {\cal G}_1$, given by \eqref{eq-2024-43}, such that $g(T) = \alpha$ for $T \in [T_1, \oo]$.
			
      \emph{proof}: For any given $\alpha_0 \wedge \alpha_1 < \alpha < \alpha_0 \vee \alpha_1$, we know that there exists $T \in [T_1, \oo]$ such that $g(T) = \alpha$ by Claim 1, see Figure \ref{Fig2}. If $F(x)$ is discontinuous at $x^\ast = r(T)$, then $F(t^\ast)-) < 1$ and $F(t^\ast)) = 1$. For any $p \geq \alpha$, we have $\VaR_p(F) = t^\ast \leq \VaR_p\big (G^{(1)}_T\big )$. Then \eqref{eq-2024-11} holds. If $F(x)$ is continuous at $t^\ast= r(T)$, then, in view of $\VaR_\alpha\big (G^{(1)}_T\big ) = \VaR_\alpha(F) =t^\ast$, we have
	  \begin{eqnarray}\label{eq-2024-12}
		& & \TVaR_\alpha\big (G^{(1)}_T\big ) - \TVaR_\alpha(F) \nonumber \\[3pt]
           & & \qquad = \VaR_\alpha\big (G^{(1)}_T\big ) + \frac{1}{1-\alpha} \int_{\VaR_\alpha(G^{(1)}_T)}^\oo \Gbar^{(1)}_T(x) \d x - \VaR_\alpha(F) - \frac{1}{1-\alpha}  \int_{\VaR_\alpha(F)}^\oo \Fbar(x) \d x \nonumber \\[3pt]
               & &\qquad = \frac{1}{1-\alpha} \(\int_{t^\ast}^\oo \Gbar^{(1)}_T(x) \d x - \int_{t^\ast}^\oo \Fbar(x) \d x\) \nonumber \\[3pt]
                     & &\qquad=\frac{1}{1-\alpha}\(\int_0^{F^{-1}(\alpha)}\Fbar(x)\d x-\int_0^{\(G^{(1)}_T\)^{-1}(\alpha)} \Gbar^{(1)}_T(x) \d x\),
      \end{eqnarray}
     where the last equality follows from $\mu(F) = \mu(G^{(1)}_T) = 1$. To show $\TVaR_\alpha\big (G^{(1)}_T\big ) \geq \TVaR_\alpha(F)$, combining \eqref{eq-2024-12} and Proposition \ref{pro-2-1} (iv), it suffices to consider
	 \begin{equation}\label{eq-2024-13}
			\phi_F(\alpha) \geq \phi_{G^{(1)}_T}(\alpha),
	\end{equation}
    where $\phi_F$ and $\phi_{G^{(1)}_T}$ are the TTTs of $F$ and $G^{(1)}_T$, respectively. By Proposition \ref{pro-2-1} (i), $F$ is absolutely continuous and strictly increasing in $x \in \{z \geq 0, F(z) <1\}$. Then $F^{-1}(p)$ is also absolute continuous in $p \in (\alpha_o\wedge\alpha_1, \alpha_0\vee \alpha_1)$. Combining $F(t^\ast) = G^{(1)}_T(t^\ast) = \alpha$, we calculate the derivatives of TTTs for $F$ and $G^{(1)}_T$ at $\alpha$, i.e.,
	\begin{equation*}
           \phi_F'(\alpha) = \frac{1}{\lm_F(t^\ast)}\ \ {\rm and}\ \ \phi_{G^{(1)}_T}'(\alpha) = \frac{1}{\lm_{G^{(1)}_T}(t^\ast)},
	\end{equation*}
    where $\lm_F$ and $\lm_{G^{(1)}_T}$ are the corresponding failure rates of $F$ and $G^{(1)}_T$. Since $F$ up-crosses $G^{(1)}_T$ at $t^\ast$, then $\lambda_{G^{(1)}_T}(t^\ast) < \lambda_F(t^\ast)$, implying  $\phi_F'(\alpha) < \phi_{G^{(1)}_T}'(\alpha)$. If $\phi_F(\alpha) < \phi_{G^{(1)}_T}(\alpha)$, by the concavity of $\phi_F$ (see Proposition \ref{eq-2024-1} (iv)) and linearity of $\phi_{G^{(1)}_T}$, we obtain
	$$
		\phi_F(1)\leq \phi_F(\alpha) + \phi_F'(\alpha)(1 -\alpha)
			< \phi_{G^{(1)}_T}(\alpha) + \phi_{G^{(1)}_T}'(\alpha)(1 - \alpha) = \phi_{G^{(1)}_T}(1),
   $$
   which contradicts $\phi_F(1) = 1 = \phi_{G^{(1)}_T}(1)$. Thus, \eqref{eq-2024-13} holds and we complete the proof of Claim 2.
\end{itemize}
Thus, for any $\alpha \in (0, 1)$, \eqref{eq-2024-9} holds. This completes the proof of the theorem.
\end{proof}
	
Based on Theorem \ref{thm-4-4}, we now establish the worst-case RVaR bound. The subsequent analysis requires greater sophistication than that in the proof of Theorem \ref{thm-4-4}, as the parameter $\beta$ introduces significant analytical complexity -- especially in the small $\beta$ regime where its uncertainty dominates.
	
\begin{theorem}\label{thm-4-5}
For any $0\leq\alpha < \beta \le 1$, we have
\begin{equation}\label{eq-2024-14}
   \sup_{F\in {\cal G}} \RVaR_{\alpha, \beta}(F) =\sup_{F\in {\cal G}_1\cup {\cal G}_2}\RVaR_{\alpha, \beta}(F).
\end{equation}
\end{theorem}
	
\begin{proof}
It suffices to establish that for any $0<\alpha<\beta< 1$ and given $F \in {\cal G}$, there exists another $G \in {\cal G}_1 \cup {\cal G}_2$ such that
\begin{equation}\label{eq-2024-15}
	\RVaR_{\alpha, \beta}(F) \leq \RVaR_{\alpha, \beta}(G).
\end{equation}
		
For any given $F \in {\cal G}$, define $G_{T_1} \in {\cal G}_1$ as \eqref{eq-2024-42} and $G^{(1)}_\oo \in {\cal G}_1$ as \eqref{eq-2024-10}. Recall the definitions of $\alpha_1$ and $\alpha_0$ from the proof of Theorem \ref{thm-4-4}. That is, $\alpha_1 = F(t_1) = G_{T_1}(t_1)$, where $t_1 \in (0, T_1)$ such that $F$ up-crosses $G_{T_1}$, and $\alpha_0 = G^{(1)}_\oo(t_0)$, where $F$ up-crosses $G^{(1)}_\oo$ at $t_0 \in (\Delta_0, \oo)$, see also Figure \ref{Fig1}. Moreover, we have $F(T_1) < 1$. Otherwise $F$ up-crosses $G_{T_1}$ only once. This contradicts the equality of the first two moments by Lemma \ref{cut}. Furthermore, we also have $G_\oo \in {\cal G}_2$, since $\Gbar^{(1)}_\oo = \Gbar^{(2)}_{\Delta_0} \in {\cal G}_2$. To see this, when substituting $\Gbar_{\Delta_0}^{(2)}$ into \eqref{eq-2024-45}, we get $k_1 = 0$ and $k_2 = (\mu_2 - 1)^{-1/2}$. We can assert that $\Delta_0 = T_0$ and $b = k_2$ by the equally first two moments. In order to better reflect our proof, we always denote $G^{(2)}_{T_0} = G^{(1)}_\oo$ below and $\alpha_0 = G^{(2)}_{T_0}(t_0)$.
		
Below, we show how to construct $G \in {\cal G}_1 \cup {\cal G}_2$ such that \eqref{eq-2024-15} holds for any given $F \in {\cal G}$ with $F \neq G_{T_1}$ and $F \neq G^{(2)}_{T_0}$. We consider three cases:
		
\underline{\emph{Case 1:}}\ Suppose $\alpha \geq \alpha_0$. In this case, we have $F(x) > G^{(2)}_{T_0}(x)$ for $x > t_0$ by the convexity of $\Lambda_F$, see Figure \ref{Fig1}. It is obviously to see that, $\VaR_p(F) \leq \VaR_p(G^{(2)}_{T_0})$ for any $p > \alpha_0$. Finally, we have
\begin{equation*}
	\RVaR_{\alpha, \beta}(F) \leq \RVaR_{\alpha, \beta}(G^{(2)}_{T_0}).
\end{equation*}
		
\underline{\emph{Case 2:}}\ Suppose $\alpha \leq \alpha_1$. In this case, denote $h_1 := F(T_1) < 1$ by the previous analysis. Moreover, $F$ down-crosses first and then up-crosses $G^{(2)}_{T_0}$. Denote $F$ down-crosses $G^{(2)}_{T_0}$ at $u(T_0) \in (T_0, \oo)$ and $G^{(2)}_{T_0}(u(T_0)) = F(u(T_0)) =: h_0$. For any $T_0 < T < T_1$. Let's consider the case where $F$ intersects with $G^{(2)}_T$ defined in \eqref{eq-2024-44}. If $F$ is greater than or equal to $G^{(2)}_T$ on $(0, T)$, then $F$ will down-cross first and then up-cross $G^{(2)}_T$ on $(T, \oo)$. If $F$ up-crosses $G^{(2)}_T$ on $(0, T)$, then $F$ will down-cross $G^{(2)}_T$ or $F$ will down-cross first then up-cross $G^{(2)}_T$ on $(T, \oo)$. Anyway, we can always ensure that $F$ down-crosses $G^{(2)}_T$ at $(T, \oo)$, and denote this point as $u(T) \in (T, \oo)$. Denote
$$
   h(T):= G^{(2)}_T(u(T)) = F(u(T)).
$$
Note that, $h(T_0) = h_0$. Moreover, by the proof process of Theorem 3.1 in \cite{BM64}, then $u(T)$ is continuous in $T \in (T_0, T_1)$ and $\lim_{T \to T_1} u(T) = T_1$. Thus, $h(T_1) = h_1$. For graphic interpretation, see Figure \ref{Fig3}.
		
Below, we show that $F$ must be continuous at $u(T)$ for $T \in [T_0, T_1]$. Assume that $F$ is not continuous at $u(T)$, then $F(u(T)) = 1$ and $F(u(T)-) < 1$. After $u(T)$, $F$ cannot up-crosses $G^{(2)}_T$ again, this contradicts the equally first two moments. Thus, $h(T)$ is continuous for $T \in [T_0, T_1]$. Below, we consider three subcases:
\begin{itemize}
  \item Assume $\beta \geq h_1$. In this subcase, by Theorem \ref{thm-4-4}, we have
	\begin{equation*}
		\TVaR_{\alpha}(F) \leq \TVaR_{\alpha}(G_{T_1}).
	\end{equation*}
	Moreover, we have $\VaR_p(F) \geq \VaR_{h_1}(F) = T_1 = \VaR_p(G_{T_1})$ for all $p \geq \beta$. Thus,
	\begin{align*}
         \RVaR_{\alpha, \beta}(G_{T_1}) &= \frac{1}{\beta - \alpha} \[(1-\alpha)\TVaR_{\alpha}(G_{T_1}) - (1-\beta)\TVaR_{\beta}(G_{T_1})\] \\
             & \geq \frac{1}{\beta - \alpha} \[(1-\alpha)\TVaR_{\alpha}(F) - (1-\beta)\TVaR_{\beta}(F)\]
				= \RVaR_{\alpha, \beta}(F).
	\end{align*}
			
  \item Assume $\beta \leq h_0$. In this subcase, we have $G^{(2)}_{T_0}(x) \leq F(x)$ for any $x \in (0, u(T_0))$. Then, $\VaR_p(F) \leq \VaR_p(G^{(2)}_{T_0})$ for $p \in (0, h_0)$. Thus,
	  $$	\RVaR_{\alpha, \beta}(G^{(2)}_{T_0}) \geq \RVaR_{\alpha, \beta}(F).   $$
			
  \item Assume $h_0 \wedge h_1 < \beta < h_0 \vee h_1$. In this subcase, there exists $G^{(2)}_T \in {\cal G}_2$  such that $F$ down-crosses $G^{(2)}_T$ at $u(T)$ and $F(u(T)) = G^{(2)}_T(u(T)) = \beta$, where $G^{(2)}_T$ is given by \eqref{eq-2024-44}. If $G^{(2)}_T(x) \leq F(x)$ for any $x \in (0, u(T))$, then we have $\VaR_p(F) \leq \VaR_p(G^{(2)}_T)$ for any $p \in (\alpha, \beta)$. Thus, $\RVaR_{\alpha, \beta}(G^{(2)}_T) \geq \RVaR_{\alpha, \beta}(F)$. Otherwise, there exists $v_1(T)$ and $v_2(T)$ with $v_1(T) < u(T) < v_2(T)$ such that $F$ first up-crosses, then down-cross and up-crosses $G_T$ at $v_1(T)$, $u(T)$, and $v_2(T)$, respectively. For graphic interpretation, see Figure \ref{Fig4}. Note that, $v_2(T)$ may be infinity. Denote
      \begin{equation} \label{eq-20250625}
	   	\psi(t) := \int_0^t \[\Fbar(x) - \Gbar^{(2)}_T(x)\] \d x.
     \end{equation}
     Note that $\psi(t)$ is increasing on $(0, v_1(T))$ first, and then decreasing, increasing and decreasing on  $(v_1(T), u(T))$, $(u(T), v_2(T))$ and $(v_2(T), \oo)$, respectively. Note that, $\psi(\oo) = 0$ and $\psi(0) = 0$. If $\psi(u(T)) \geq 0$, then $\psi(t) \geq 0$ for any $t > 0$. Thus,
	 \begin{align*}
            \int_0^\oo 2x\[\Fbar(x)-\Gbar^{(2)}_T(x)\] \d x &= 2\int_0^\oo \int_0^x \[\Fbar(x)-\Gbar^{(2)}_T(x)\]\d y\d x\\[3pt]
				& = 2 \int_0^\oo \int_y^\oo \[\Fbar(x) - \Gbar^{(2)}_T(x)\] \d x \d y\\[3pt]
       	& = -2 \int_0^\oo \int_0^y \[\Fbar(x) - \Gbar^{(2)}_T(x)\] \d x \d y < 0.
	\end{align*}
    This contradicts $\mu_2(F) = \mu_2(G^{(2)}_T)$. Thus, $\psi(u(T)) < 0$, i.e.,
    \begin{equation} \label{eq-250626}
        \RVaR_{0, \beta}(G^{(2)}_T) > \RVaR_{0, \beta}(F).
    \end{equation}
    Denote $\alpha' := F(v_1(T))$. If $\alpha \geq \alpha'$, then $\VaR_p(F) \leq \VaR_p(G^{(2)}_T)$ for $p \in (\alpha, \beta)$ since $G^{(2)}_T(x) \leq F(x)$ for $x \in (v_1(T), u(T))$. Thus $\RVaR_{\alpha, \beta}(G^{(2)}_T) > \RVaR_{\alpha, \beta}(F)$. If $\alpha < \alpha'$, then $\VaR_p(F) >\VaR_p(G^{(2)}_T)$ for $p \in (0, \alpha')$. Thus, in view of \eqref{eq-250626}, we have
	\begin{align*}
	  \RVaR_{\alpha, \beta}(G^{(2)}_T) &= \frac{1}{\beta - \alpha} \[\beta\, \RVaR_{0, \beta}(G^{(2)}_T)
				- \alpha\, \RVaR_{0, \alpha}(G^{(2)}_T)\] \\
			&  \geq \frac{1}{\beta - \alpha} \[\beta\RVaR_{0, \beta}(F) - \alpha\RVaR_{0, \alpha}(F)\]
				= \RVaR_{\alpha, \beta}(F).
	\end{align*}
\end{itemize}
				
 \begin{figure}[htbp]
	\centering
	\begin{minipage}[b]{0.45\textwidth}
		\begin{tikzpicture}
			\tikzset{
				box/.style ={
	     			circle,
					minimum width =1.5pt,
					minimum height =1.5pt,
					inner sep=1.5pt,
					draw=black,
					fill=white
				}
			}
			\draw[->,thick] (0,0) -- (7,0) node[right] {$x$};
			\draw[->,thick] (0,0) -- (0,8) node[above] { };			
			\draw[domain=0:6] plot (\x,{0.2*\x*\x})node[above] {$\Lambda_F$};
			\draw[domain=0:4,color=red] plot (\x,{0.6*\x})node[below right] {$\Lambda_{{G}_{T_1}}$};
			\draw[domain=1:6,color=blue] plot (\x,{1.2*(\x-1)})node[right] {$\Lambda_{{G}^{(2)}_{T_0}}$};
			\node at (3,1.77) { $\bullet$};
			\node at (4.2,1.5) {\footnotesize $(r(T_1), -\ln(1-\alpha_1))$};
			\node[box] at (3.95,2.35){};
			\node at(1.26795,0.321539){ $\bullet$};
			\node at (2.7,0.25) {\footnotesize $(u(T_0), -\ln(1-h_0))$};
			\node at(4.73205,4.47846){ $\bullet$};
			\node at (6.1,4.47846) {\footnotesize $(r(\oo), -\ln(1-\alpha_0))$};
			\node at (3.95,3.14) { $\bullet$};
			\node at (5.4,3.14) {\footnotesize $(u(T_1), -\ln(1-h_1))$};
			\draw  [dash pattern={on 1.5pt off 2pt}] (3.95,0) -- (3.95,3.13);
			\node at (3.95,-0.3) {\footnotesize $T_1$};
			\node at (1.0,0) { $\bullet$};
			\node at (1.0,-0.3) {\footnotesize $T_0$};
		\end{tikzpicture}
		\caption{The graphs of $\Lambda_F$, $\Lambda_{{G}^{(2)}_{T_0}}$ and $\Lambda_{{G}_{T_1}}$.}
		\label{Fig3}
		\end{minipage}
		\hfill
		\begin{minipage}[b]{0.45\textwidth}
		\begin{tikzpicture}
			\tikzset{
				box/.style ={
					circle,
					minimum width =1.5pt,
					minimum height =1.5pt,
					inner sep=1.5pt,
					draw=black,
					fill=white
				}
			}
		\draw[->,thick] (0,0) -- (7.5,0) node[right] {$x$};
		\draw[->,thick] (0,0) -- (0,8.5) node[above] { };			
		\draw[domain=0:6.5] plot (\x,{0.2*\x*\x})node[above] {$\Lambda_F$};
		\draw[domain=0:2.8,color=blue] plot (\x,{0.4*\x});
		\draw[domain=2.8:6.5,color=blue] plot (\x,{1.9*(\x-2.8)+1.12})node[right] {$\Lambda_{G^{(2)}_T}$};;
		\node at (3.49,2.45) { $\bullet$};
		\node at (6,7.2) { $\bullet$};
		\node at (2,0.8) { $\bullet$};
		\draw  [dash pattern={on 1.5pt off 2pt}] (3.49,0) -- (3.49,2.45);
		\draw  [dash pattern={on 1.5pt off 2pt}] (6,0) -- (6,7.2);
		\draw  [dash pattern={on 1.5pt off 2pt}] (2,0) -- (2,0.8);
        \node at (-0.3,2.6) {\footnotesize \rotatebox{90}{$-\ln(1\!-\!h(T))$}};
       	\draw  [dash pattern={on 1.5pt off 2pt}] (0,0.8) -- (1.76,0.8);
        \node at (-0.3,0.6) {\footnotesize \rotatebox{90}{$-\ln(1\!-\!\alpha')$}};
        \node at (-0.3,7.2) {\footnotesize \rotatebox{90}{$-\ln(1\!-\!\beta')$}};
		\node at (3.49,-0.3) {\footnotesize $u(T)$};
		\node at (6,-0.3) {\footnotesize $v_2(T)$};
		\node at (2,-0.3) {\footnotesize $v_1(T)$};
		\node at (3.49,6) {\footnotesize $T_0<T<T_1$};
		\node at (2.8,1.12) { $\bullet$};
        \draw  [dash pattern={on 1.5pt off 2pt}] (0,7.2) -- (6,7.2);
        \draw  [dash pattern={on 1.5pt off 2pt}] (0,2.45) -- (3.4,2.45);
		\draw  [dash pattern={on 1.5pt off 2pt}] (2.8,0) -- (2.8,1.12);
		\node at (2.8,-0.3) {\footnotesize $T$};
		\end{tikzpicture}
		\caption{The graphs of $\Lambda_F$ and $\Lambda_{{G}^{(2)}_{T}}$.}
		\label{Fig4}
		\end{minipage}
 \end{figure}
		
\underline{\emph{Case 3:}}\ Suppose $\alpha_0 \wedge \alpha_1 < \alpha < \alpha_0 \vee \alpha_1$. In this case, for given $\alpha$, it is easy to see that there exists some point $T \in [T_1, \oo]$ such that  $g(T) = \alpha$ by Claim 1 in Theorem \ref{thm-4-4}; see Figure \ref{Fig2}. If $F(x)$ is not continuous at $r(T)$, then $F(r(T)-) < 1$ and $F(r(T)) = 1$. For any $p \geq \alpha$, we have $\VaR_p(F) = r(T) \leq \VaR_p(G^{(1)}_T)$. Then \eqref{eq-2024-15} holds. If $F(x)$ is continuous at $r(T)$, then by $\VaR_\alpha(G^{(1)}_T) = \VaR_\alpha(F) = r(T)$. Below, we consider two subcases:
\begin{itemize}
    \item Suppose $\VaR_\beta(F) \leq T$. We have $F(x) \geq G^{(1)}_T(x)$ for any $x \in [r(T), \VaR_\beta(F)]$. Then $\VaR_p(F) \leq \VaR_p (G^{(1)}_T)$ for any $p \in [\alpha, \beta]$. Then \eqref{eq-2024-15} holds.
			
   \item Suppose $\VaR_\beta(F) > T$. In this case, we have $\beta > F(T) \geq F(T-) \geq G^{(1)}_T(T-)$.			 Then, it is easy to see that $\VaR_p(F) \geq \VaR_\beta(F) > T = \VaR_p(G^{(1)}_T)$. Thus,
       $$
              \TVaR_\beta(F) = \frac{1}{1-\beta} \int^1_{\beta} \VaR_p(F) \d p > \frac{1}{1-\beta}  \int^1_{\beta} \VaR_p(G^{(1)}_T) \d p = \TVaR_\beta(G^{(1)}_T).
	   $$
       Note that by the proof of Theorem \ref{thm-4-4}, we have $\TVaR_\alpha(G^{(1)}_T) \geq \TVaR_\alpha(F)$. Thus,
	  \begin{align*}
          \RVaR_{\alpha, \beta}(G^{(1)}_T) &= \frac{1}{\beta - \alpha} \[(1-\alpha)\TVaR_{\alpha}(G^{(1)}_T) - (1-\beta)\TVaR_{\beta}(G^{(1)}_T)\] \\
              & \geq \frac{1}{\beta - \alpha} \[(1-\alpha)\TVaR_{\alpha}(F) - (1-\beta)\TVaR_{\beta}(F)\]  =\RVaR_{\alpha, \beta}(F)
  	  \end{align*}
\end{itemize}
Thus, for any $\alpha \in (0, 1)$, \eqref{eq-2024-14} holds.
\end{proof}
	
A comparison of Theorems \ref{thm-4-4} and \ref{thm-4-5} indicates that, within the context of worst-case TVaR, the analysis of the worst-case distribution must be in ${\cal G}_1$ exclusively. Conversely, within the framework of worst-case RVAR, the worst-case distribution may be observed not only in ${\cal G}_1$, but also in ${\cal G}_2$. This phenomenon can be attributed to the observation that, in scenarios where the parameters $\alpha$ and $\beta$ are small, i.e. $\alpha \leq \alpha_1$ and $\beta \leq h_0$ (see Case 2 in the proof of Theorem \ref{thm-4-5}), the worst-case distribution is attained in ${\cal G}_2$.
	
Finally, we examine the best-case RVaR scenario. Analogous to Theorem \ref{thm-4-5}, our analysis demonstrates that the optimal RVaR under the first two moments and IFR distribution constraints is likewise attained within the sets ${\cal G}_1$ and ${\cal G}_2$.
	
\begin{theorem}\label{thm-4-6}
For any $0<\alpha< \beta\le 1$, we have
\begin{equation}  \label{eq-2024-16}
      \inf_{F \in {\cal G}} \RVaR_{\alpha, \beta}(F) = \inf_{F \in {\cal G}_1 \cup {\cal G}_2} \RVaR_{\alpha, \beta}(F).
\end{equation}
\end{theorem}
	
\begin{proof}
It suffices to show that for given $F \in {\cal G}$, there exists $G \in {\cal G}_1 \cup {\cal G}_2$ such that
\begin{equation}\label{eq-2024-17}
		\RVaR_{\alpha, \beta}(G) \leq \RVaR_{\alpha, \beta}(F).
\end{equation}
		
From Theorems \ref{thm-4-4} and \ref{thm-4-5}, $F$ up-crosses $G^{(1)}_T \in {\cal G}_1$ at $r(T) \in (\Delta, T)$ when $T \geq T_1$, where $G^{(1)}_T$ is given by \eqref{eq-2024-43}. Besides, $g(T) = G^{(1)}_T(r(T)) = F(r(T))$, $\alpha_1 = g(T_1)$ and $\alpha_0 = g(\oo)$. Moreover, $F$ down-crosses $G^{(2)}_T \in {\cal G}_2$ at $u(T) \in (T, \oo)$ when $T \in (T_0, T_1)$. In addition, $h(T) = G^{(2)}_T(u(T)) = F(u(T))$, $h_1 = h(T_1) < 1$ and $h_0 = h(T_0)$.	Below, we show how to construct $G \in {\cal G}_1 \cup {\cal G}_2$ such that \eqref{eq-2024-17} holds for any given $F \in {\cal G}$ with $F \neq G_{T_1}$ and $F \neq G^{(2)}_{T_0}$. We consider three cases:
		
\underline{\emph{Case 1:}}\ Suppose $\beta \leq \alpha_1$. In this case, we have $G_{T_1}(x) \geq F(x)$ for $x \in [0, r(T_1)]$ by the convexity of $\Lambda_F$, see Figure \ref{Fig1}. It is obviously to see that, $\VaR_p(F) \geq \VaR_p(G_{T_1})$ for any $p \in [\alpha, \beta]$, implying
\begin{equation*}
	\RVaR_{\alpha, \beta}(F) \geq \RVaR_{\alpha, \beta}(G_{T_1}).
\end{equation*}
		
\underline{\emph{Case 2:}}\ Suppose $\beta \geq \alpha_0$. In this case, we need to consider the following three subcases, see Figure \ref{Fig3}.
\begin{itemize}
   \item Assume $\alpha \leq h_0$. In this subcase, we have $G^{(2)}_{T_0}(x) \leq F(x)$ on $x \in [0, u(T_0)]$. Then $\VaR_p(F) \leq \VaR_p(G^{(2)}_{T_0})$ on $p \in [0, \alpha]$. Therefore, $\TVaR_\alpha(F) \geq \TVaR_\alpha(G^{(2)}_{T_0})$ since $\mu(F) = \mu(G_{T_0})$. Moreover, we also have $\TVaR_\beta(F) \leq \TVaR_\beta(G^{(2)}_{T_0})$ by Case 1 in Theorem 4.4 and $G^{(2)}_{T_0} = G_\oo$. Note that
       $$
             \RVaR_{\alpha, \beta}(F) = \frac{1}{\beta-\alpha} \int_\alpha^\beta \VaR_p(F) \d p = \frac{1}{\beta-\alpha} \[(1-\alpha) \TVaR_\alpha(F) - (1-\beta) \TVaR_\beta(F)\]
	   $$
	   Thus, \eqref{eq-2024-17} holds.
			
   \item Assume $\alpha \geq h_1$. In this subcase, we have $h_1 = F(T_1) > \lim_{t \uparrow T_1} G_{T_1}(t)$. 		 Then, $\VaR_p(G_{T_1}) = T_1 = \VaR_{h_1}(F) \leq \VaR_p(F)$ for $p \geq \alpha$. Thus,
	  $$	\RVaR_{\alpha, \beta}(G_{T_1}) \leq \RVaR_{\alpha, \beta}(F).   $$
			
  \item Assume $h_0 \wedge h_1 < \alpha < h_0 \vee h_1$. In this subcase, there exists $G^{(2)}_T \in {\cal G}_2$  such that $F$ down-crosses $G^{(2)}_T$ at $u(T)$ and $F(u(T)) = G^{(2)}_T(u(T)) = \alpha$, where $T \in (T_0, T_1)$ and $G^{(2)}_T$ is given by \eqref{eq-2024-44}. Following the similar analysis of Case 2 in the proof of Theorem \ref{thm-4-5}, we know that either $G^{(2)}_T(x) \geq F(x)$ for any $x \in (u(T), \oo)$, or $F$ up-cross only once $G^{(2)}_T$ at $v_2(T) \in (u(T), \oo)$. For graphic interpretation, see Figure \ref{Fig4}. In the former type, we have $\VaR_p(F) \geq \VaR_p(G^{(2)}_T)$ for any $p \in [\alpha, \beta]$, so \eqref{eq-2024-17} holds. In the later type, denote
      $$
           \beta' := G^{(2)}_T(v_2(T)) = F(v_2(T)).
      $$
      If $\beta \leq \beta'$, then $G^{(2)}_T(x) \geq F(x)$ for any $x \in (u(T), v_2(T))$, and $\VaR_p(F) \geq \VaR_p(G^{(2)}_T)$ for any $p \in [\alpha, \beta]$, so \eqref{eq-2024-17} holds. If $\beta > \beta'$, then $G^{(2)}_T(x) \leq F(x)$ for any $x \in (v_2(T), \oo)$, and $\VaR_p(F) \leq \VaR_p(G^{(2)}_T)$ for any $p \in [\beta, 1]$. Therefore, $\TVaR_\beta(F) \leq \TVaR_\beta(G^{(2)}_T)$. Moreover, we also have $\RVaR_{0, \alpha}(F) \leq \RVaR_{0, \alpha}(G^{(2)}_T)$ by the similar Case 2 in Theorem \ref{thm-4-5}. Thus, \eqref{eq-2024-17} holds.
\end{itemize}
		
\underline{\emph{Case 3:}}\ Suppose $\alpha_0 \wedge \alpha_1 < \beta < \alpha_0 \vee \alpha_1$. In this case, it is easy to see that there exists some point $T \in (T_1, \oo)$ such that $g(T) = \beta$ by Claim 1 in Theorem \ref{thm-4-4}. Following the similar analysis of Case 3 in Theorem \ref{thm-4-4}, we know that either $G^{(1)}_T(x) \geq F(x)$ for any $x \in (0, r(T))$, or $F$ down-cross $G^{(1)}_T$ at $t' \in (0, r(T))$, see Figure \ref{Fig2}. In the former type, we have $\VaR_p(F) \geq \VaR_p(G^{(1)}_T)$ for any $p \in [\alpha, \beta]$, so \eqref{eq-2024-17} holds. In the later type, denote
$$
      \alpha' := G^{(1)}_T(t') = F(t').
$$
If $\alpha \geq \alpha'$, then $G^{(1)}_T(x) \geq F(x)$ for any $x \in (t', r(T))$, and $\VaR_p(F) \geq \VaR_p(G^{(1)}_T)$ for any $p \in [\alpha, \beta]$, so \eqref{eq-2024-17} holds. If $\alpha < \alpha'$, then $G^{(1)}_T(x) \leq F(x)$ for any $x \in (0, t')$, and $\VaR_p(F) \leq \VaR_p(G^{(1)}_T)$ for any $p \in [0, \alpha]$. Therefore, $\TVaR_\alpha(F) \geq \TVaR_\alpha(G^{(1)}_T)$ since $\mu(F) = \mu(G^{(1)}_T)$. Moreover, we also have $\TVaR_\beta(F) \leq \TVaR_\beta(G^{(1)}_T)$ by Claim 2 in Theorem \ref{thm-4-4}. Thus, \eqref{eq-2024-17} holds.
		
Therefore, by considering the above three cases, we have \eqref{eq-2024-14}. This completes the proof of the theorem.
\end{proof}

\section{Applications to stop-loss and limited loss transforms}\label{sec5}

In this section, we are going to apply the previous results to the extreme-case stop-loss transform and extreme-case limited loss transform. For any $t \in \R_+$ and a loss random variable $X_F$ with distribution $F\in {\cal M}$, the stop-loss transformation $(X_F-t)_+$ and the limited loss random variable $X_F \wedge t$ are two of the most commonly used transformations in insurance and finance. Note that
$$
     \sup_{F\in {\cal P}} \E\[\(X_F-t\)_+\] = \E\[X_F\]- \inf_{F\in {\cal P}} \E\[X_F \wedge t\],
$$
and
$$
     \inf_{F \in {\cal P}} \E\[\(X_F-t\)_+\] = \E\[X_F\] - \sup_{F \in {\cal P}} \E\[X_F \wedge t\],
$$
where $\cal P$ is some ambiguity set of distributions. We only need to consider the extreme-case stop-loss transform based on the above observation. In this section, we consider the following questions to illustrate our previously established results: for any given $t \in \R_+$,
\begin{equation}\label{eq5-1}
	\sup_{F \in {\cal P}} \E\[\(X_F - t\)_+\],\ {\rm and}\ \inf_{F \in {\cal P}} \E\[(X_F - t)_+\],
\end{equation}
where ${\cal P}$ is ${\cal F} = \left\{F \in {\cal M}: F {\rm\ is\ IFR},\ \mu(F) = 1 \right\}$ or ${\cal G} = \left\{F \in {\cal M}: F {\rm\ is\ IFR},\ \mu(F) = 1,\ \mu_2(F) = \mu_2\right\}$. We first derive the result for the worst-case scenario by applying the reverse TVaR optimization formula, originally proposed by \cite{GJZ24} and later generalized by \cite{GZGH24}.

\begin{proposition}\label{pro5-1}
For $t \in \R_+$, we have
\begin{eqnarray*}
	& & \sup_{F \in {\cal F}} \E\[\(X_F-t\)_+\]  =  \exp\{-t\},\\
    & & \sup_{F \in {\cal G}} \E\[\(X_F-t\)_+\]  = \sup_{\alpha \in [0,1]} \left\{(1-\alpha) \left(\sup_{F \in {\cal G}_1} \TVaR_\alpha(F) - t\right)\right\},
\end{eqnarray*}
where ${\cal G}_1$ is defined by \eqref{distributions}.
\end{proposition}

\begin{proof}
It is well known that
\begin{eqnarray*}
	\E\[\(X_F - t\)_+\] = \sup_{\alpha \in [0,1]} \left\{(1-\alpha) \[\TVaR_\alpha(F) - t\]\right\}.
\end{eqnarray*}
Thus,
\begin{align}\label{eq5-2}
   \sup_{F \in {\cal P}} \E\[\(X_F-t\)_+\] &= \sup_{F \in {\cal P}} \sup_{\alpha \in [0,1]} \left\{(1-\alpha) \(\TVaR_\alpha(F) - t\)\right\} \nonumber \\
        &= \sup_{\alpha\in [0,1]}\sup_{F \in {\cal P}}\left\{(1-\alpha) (\TVaR_\alpha(F)-t)\right\}\nonumber \\
       &=\sup_{\alpha \in [0,1]} \left\{(1-\alpha) \(\sup_{F\in {\cal P}} \TVaR_\alpha(F)- t\)\right\}.
\end{align}
When ${\cal P} = {\cal F}$, we have
\begin{align*}
    \sup_{F \in {\cal F}} \E\[\(X_F-t\)_+\] &= \sup_{\alpha \in [0,1]} \left\{(1-\alpha) \left(\sup_{F \in {\cal F}} \TVaR_\alpha(F) - t\right)\right\} \\
	 &=\sup_{\alpha \in [0,1]} \left\{(1-\alpha) (1 - \ln(1-\alpha) - t)\right\}= \exp\{-t\},
\end{align*}
where the second equality follows \eqref{eq-250628}. When ${\cal P} = {\cal G}$, we have
\begin{align*}
    \sup_{F\in {\cal G}} \E\[\(X_F-t\)_+\] &= \sup_{\alpha \in [0,1]} \left\{(1-\alpha) \left(\sup_{F \in {\cal G}} \TVaR_\alpha(F) - t\right)\right\} \\
           &= \sup_{\alpha \in [0,1]} \left\{(1-\alpha) \(\sup_{F \in {\cal G}_1} \TVaR_\alpha(F) -t\)\right\},
\end{align*}
where the second equality follows from Theorem \ref{thm-4-4}.
\end{proof}

\begin{proposition}\label{pro5-2}
For $t \in \R_+$, we have
\begin{eqnarray*}
      \inf_{F \in {\cal F}} \E\[\(X_F-t\)_+\]  =  (1-t)_+,\ \  {\rm\ and\ } \ \
      \inf_{F \in {\cal G}} \E\[\(X_F-t\)_+\]  = \inf_{F\in {\cal G}_1 \cup {\cal G}_2} \E\[\(X_F-t\)_+\],
\end{eqnarray*}
where ${\cal G}_1$ and ${\cal G}_2$ are defined by \eqref{distributions}.
\end{proposition}

\begin{proof}
When ${\cal P} = {\cal F}$, we have $\delta_1 \leq_{\rm cx} F$. Thus
$$
		\inf_{F \in {\cal F}} \E\[\(X_F - t\)_+\] =(1 - t)_+.
$$
When ${\cal P} = {\cal G}$, we need to establish that for any given $F \in {\cal G}$, there exists $G \in {\cal G}_1 \cup {\cal G}_2$ such that
\begin{equation}\label{eq-2024-53}
		\E\[\(X_G - t\)_+\] \leq \E\[\(X_F - t\)_+\].
\end{equation}
Following the same notations as those in the proof of Theorem \ref{thm-4-5}, we know that $F$ down-crosses $G^{(2)}_T \in {\cal G}_2$ at $u(T) \in (T, \oo)$ when $T \in (T_0, T_1)$. Moreover, $u(T)$ is continuous over $T \in (T_0, T_1)$ and $\lim_{T \to T_1} u(T) = T_1$. Below, we consider three cases:
		
\underline{\emph{Case 1:}}\ Suppose that $t \leq u(T_0)$. We know that $F(x) \leq G^{(2)}_{T_0}(x)$ for $x \in (0, u(T_0))$ by the convexity $\Lambda_F$. Then
\begin{equation*}
	\int_0^t \overline{G}_{T_0}^{(2)}(x) \d x \leq \int_0^t \Fbar(x) \d x.
\end{equation*}
Combining $\mu(F) = \mu(G_{T_0}) = 1$, we have
\begin{equation*}
   \E\[\(X_F - t\)_+\] = \int_t^\oo \Fbar (x) \d x \geq \int_t^\oo \overline{G}_{T_0}^{(2)}(x) \d x = \E\[\(X_{G_{T_0}} - t\)_+\].
\end{equation*}
		
\underline{\emph{Case 2:}}\ Suppose that $t \geq u(T_1) = T_1$. Note that $\essinf\(G_{T_1}\) = T_1$ for $G_{T_1} \in {\cal G}_1$. Then
\begin{equation*}
	\E\[\(X_F - t\)_+\] \geq 0 = \E\[\(X_{G_{T_1}} - t\)_+\].
\end{equation*}
		
\underline{\emph{Case 3:}}\ Suppose that $u(T_0) \wedge u(T_1) < t < u(T_0) \vee u(T_1)$. From the proof of Theorem \ref{thm-4-5}, we know that there exists $G_T^{(2)} \in {\cal G}_2$ such that $F$ down-crosses $G_T^{(2)}$ at $u(T) = t$ and
\begin{equation*}
	\psi(u(T)) = \int_0^t \[\Fbar(x) - \Gbar_T^{(2)}(x)\] \d x<0,
\end{equation*}
where $\psi(x)$ is defined by \eqref{eq-20250625}. Then
\begin{align*}
       \E\[\(X_F - t\)_+\] - \E\[\(X_{G_T^{(2)}} - t\)_+\] &= \int_t^\oo \[\Fbar(x) -\Gbar_T^{(2)}(x)\]\d x\\
        & = -\int_0^t \[\Fbar(x) - \Gbar_T^{(2)}(x)\] \d x > 0.
\end{align*}
Thus, we complete the proof of the proposition.
\end{proof}

\section{Conclusions}\label{sec6}

This paper develops a constructive methodology to determine the extreme-case RVaR for IFR distributions under mean and/or variance ambiguity sets. Our approach characterizes the extremal distributions through parametric families of IFR distributions, which enables efficient computation of extreme-case RVaR bounds. While the proposed framework naturally lends itself to numerical implementation (as demonstrated by our theoretical construction), we emphasize that the current analysis focuses primarily on the theoretical development rather than numerical simulation.

We conclude with some open questions that are the focus of our ongoing research.
\begin{itemize}
 \item How to find the extreme-case RVaR under the mean/variance ambiguity sets of distributions with DFR property or other ageing notions? A promising direction to solve extreme-case RVaR under DFR ambiguity set may refer to \cite{BM64}, in which they derived the bounds for distribution functions with IFR/DFR property and given mean and/or variance.
	
 \item The existing literature has developed various ambiguity sets beyond mean/variance uncertainty sets, notably those constructed by using probability distances like the likelihood ratio \citep{LMWW22} and Wasserstein distance \citep{BPV24}. This raises an important research question. Can we develop computational method to determine the worst-case or best-case RVaR by employing probability distances between distributions while incorporating IFR or DFR constraint?
	
 \item Since RVaR constitutes a special case of distortion risk measure, we propose extending this analysis to general distortion risk measures or distortion riskmetrics \citep{PWW24} while incorporating uncertainty information about the underlying distribution's failure rate characteristics (IFR/DFR).
\end{itemize}
	
\section*{Funding}

Z. Zou gratefully acknowledges financial support from National Natural Science Foundation of China (No. 12401625), the China Postdoctoral Science Foundation (No. 2024M753074), and the Postdoctoral Fellowship Program of CPSF (GZC20232556). T. Hu gratefully acknowledges financial support from the National Natural Science Foundation of China (No. 72332007, 12371476).


\end{document}